\documentclass[12pt, draftclsnofoot, onecolumn]{IEEEtran}
\usepackage{amssymb}
\usepackage{amsmath}
\usepackage{amsfonts}
\usepackage{graphicx}
\usepackage{algorithm}
\usepackage{algorithmic}
\usepackage{epstopdf}
\usepackage{cite}
\usepackage{exscale}
\usepackage{relsize}

\usepackage{color}
\usepackage{multirow}
\usepackage{subfigure}
\usepackage{bm}

\newtheorem{theorem}{\textbf{Theorem}}

\newtheorem{lemma}{\textbf{Lemma}}

\newtheorem{proposition}{\textbf{Proposition}}

\begin{document}
\title{Analysis on Cache-enabled Wireless Heterogeneous Networks}
\author{Chenchen Yang,  Yao Yao, Zhiyong Chen, {\em Member, IEEE}, Bin Xia, {\em Senior Member, IEEE}\\
\thanks{This work has been accepted by IEEE Transactions on Wireless Communications and was partly presented in IEEE ICC 2015 \cite{myself}. C. Yang, Z. Chen, and B. Xia are with the Institute of Wireless Communications Technology (IWCT), Department of Electronic Engineering, Shanghai Jiao Tong University, Shanghai, 200240, P. R. China. Emails: \tt \{zhanchifeixiang, zhiyongchen, bxia\}@sjtu.edu.cn}
\thanks{Y. Yao is with Huawei Technologies Co., Ltd. Email: \tt \{yyao@eee.hku.hk\}}}
\maketitle
\begin{abstract}
Caching the popular multimedia content is a promising way to unleash the ultimate potential of wireless networks. In this paper, we contribute to proposing and analyzing the cache-based content delivery in a three-tier heterogeneous network (HetNet), where base stations (BSs), relays and device-to-device (D2D) pairs are included. We advocate to proactively cache the popular contents in the relays and parts of the users with caching ability when the network is off-peak. The cached contents can be reused for frequent access to offload the cellular network traffic. The node locations are first modeled as mutually independent Poisson Point Processes (PPPs) and the corresponding content access protocol is developed. The average ergodic rate and outage probability in the downlink are then analyzed theoretically. We further derive the throughput and the delay based on the \emph{multiclass processor-sharing queue} model and the continuous-time Markov process. According to the critical condition of the steady state in the HetNet, the maximum traffic load and the global throughput gain are investigated. Moreover, impacts of some key network characteristics, e.g., the heterogeneity of multimedia contents, node densities and the limited caching capacities, on the system performance are elaborated to provide a valuable insight.
\end{abstract}
\section{INTRODUCTION}
The total mobile data traffic of 2020 will increase 1000 times compared with the 2010 traffic level \cite{Cisco}. Despite the deployment of the fourth generation Long Term Evolution (LTE) and LTE-Advanced systems, the rapidly increasing wireless data demands overwhelms the throughput increase that the wireless network could afford. Various innovative throughput-increasing methods have been investigated to tackle the ever-growing wireless data challenge, such as the heterogeneous network (HetNet) \cite{5Gmultitier} and the cache-enabled content-centric network \cite{diversity,kongtao,femto2}. The state of the art is elaborated in the perspective of the two aspects respectively in the following.

 \textbf{HetNets, bringing the network closer to users}: One widely regarded as the cornerstone technology is denser node deployment, including macro base station (BS), micro BS, pico BS, femto BS and relays. Such a HetNet decreases the distance between BSs/relays and users, and thus increases the area spectral efficiency, yielding the increase of network capacity \cite{capacity1,ppp2,Userassociation}. However, the exponential growth in traffic also requires the high-speed backhaul for the connection of different type of BSs/relays and content servers \cite{femto,seven}.

  \textbf{Cache-enabled content-centric networks, bringing the content closer to users}: It has been shown that $70\%$ of the wireless traffic is from multimedia contents, e.g., videos \cite{Cisco}. Meanwhile, the multimedia contents are not accessed with the same frequency. Only a small fraction ($5-10\%$) of ``popular" contents are consumed by the majority of the users, and the less popular contents are requested by a much smaller number of users \cite{Zipf}. Moreover, following the uncannily accurate Moore's law, a tremendous amount of computing and storage capacity is held by the intelligent terminal devices and networks. As such, the popular contents can be cached in BSs, relays and devices, bringing the content closer to users. It allows users to access to the cache-enabled nodes and reduces the duplicate content transmissions, mitigating the over-the-air traffic \cite{5G2}.

  Therefore, taking advantage of the caching capability within the wireless HetNet, the content diversity and network diversity can be exploited to relieve the burden of the fast growing traffic \cite{diversity,5G1,caching1}.
\subsection{Related Work}
The role of the caching technology in the fifth generation (5G) wireless network is demonstrated in \cite{5G1,5G2}. Urs Niesen \emph {et al.} investigate a large wireless caching network with the hierarchical tree structure of transmissions, and scaling results on the capacity region are derived. An arbitrary traffic matrix and cooperative transmissions over arbitrarily long links are assumed \cite{caching1}. \cite{RAN} introduces the distributed caching at the macro BSs to improve the network capacity and reduce the video stalling. The authors of \cite{femto} advocate to set up relays with caching ability in the cellular network to reduce the access delay. The content placement scheme has received significant attention, e.g., \cite{coding} proposes a novel coded caching scheme to improve both the local and the global caching gain. \cite{geographic} considers the scenario where a user in the overlapping coverage area can connect to any of the stations covering it. The optimal caching strategy maximizing the caching hit ratio is formulated by solving the Geographic Caching Problem.  Optimal request routing and content caching are investigated in \cite{routing} to minimize the average content access delay. In \cite{niuzhisheng}, the energy consumption is minimized by appropriately pre-caching popular contents. \cite{pushing} studies how to disseminate the content via cellular caching and Wi-Fi sharing to trade off the dissemination delay and the energy cost. Based on the content popularity, the cache-based multimedia content delivery scheme is proposed and analyzed in \cite{kongtao}. Terminal users can share the received content via opportunistic local connectivity to offload the traffic of cellular links in \cite{D2D}. \cite{grid} exploits redundancy of user requests and the storage capacity of terminal devices via dividing the cell into virtual square grids.

However, in the current research of caching, the assumption of global knowledge of the stationary network topology and the node connectivity graph is critical, and the regular grid network model is too optimistic and idealistic to fully capture the randomness and complexity of node locations in the HetNet nowadays. Different from the traditional system model, lots of researches have pointed out that the node location obeys
PPP instead of regular hexagonal grid in realistic HetNet \cite{capacity1,origin,ppp2,ppp3}. Two tiers of BS locations are modeled as independent PPPs in \cite{ppp1}, where joint resource partitioning and offloading are analyzed in the HetNet. \cite{capacity2} studies the optimal node density in homogeneous and heterogeneous scenarios by modeling cellular networks with PPP. The authors in \cite{SINR} model the node locations of the multi-tier HetNet as mutually independent PPPs, and analyze the system performance in terms of the average rate. \cite{cache-enable} takes the limited backhaul into consideration to analyze the performance of the homogeneous cache-enabled small cell network, where the nodes of the small base stations are stochastically distributed. A constant service rate is assumed when the files can be found in the local cache and the downlink capacity exceeds the threshold.
In \cite{trade-off}, disjoint circular clusters are scattered based on the hard-core PP. Requesting users and cache-enabled users are distributed with two independently homogeneous PPPs. The requesting users obtain the content from cache-enabled users in the same cluster via the out-of-band device-to-device (D2D) in the cellular network.

Furthermore, the traditional fetching and reactive caching methods doesn't intelligently utilize the service characteristic such as the traffic redundancy and the content popularity. Few studies considers the scenario where the radio access network (RAN) caching and the D2D caching coexist. Meanwhile, the performance of the wireless cooperative caching HetNet is not yet fully investigated. How much performance improvement actually can be reaped via the caching technology is urgent to be answered theoretically.
\subsection{Contributions}
Towards these goals, in this paper we analyze the scheme that when the network load is off-peak, the most popular contents can be cached at the nodes via broadcasting. The BSs, relays and cache-enabled users are cooperative to transmit contents in the HetNet. The main contributions of this paper are summarized as follows:

\begin{itemize}
  \item We consider the limited caching ability of both relays and parts of the users. Popular contents are cached when the network is off-peak. Besides the cellular communication, there exists the local content sharing links from the cache-enabled user to the users. When a user triggers a request, it can be responded by BSs, relays or the cache-enabled users.
  \item We model the node locations (BSs, relays and users) of the three-tier HetNet (BSs-users, relays-users, users with caching ability-users) as mutually independent PPPs. The content access protocol is then proposed, based on which the tier association priority is formulated.
  \item We derive analytical expressions of the average ergodic rate and outage probability for users in different Cases. Then with the modeling of the request arrival and departure process at the service node as a \emph{multiclass processor-sharing queue}, the throughput and delay of different classes are further analyzed based on the continuous-time Markov process.
  \item We propose the \emph{steady ruler} and the critical point for the HetNet to keep steady, according to which the throughput and the maximum traffic load over the entire network are then evaluated. Moreover, impacts of the cache-enabled users, content popularity and the limited storage capacity on the network performance are analyzed.
\end{itemize}

The remainder of the paper is organized as follows:
In Section \uppercase\expandafter{\romannumeral2},  we formulate the three-tier HetNet architecture and elaborate the tier association priority based on the content access protocol. The average ergodic rate and outage probability are derived in Section \uppercase\expandafter{\romannumeral3} and Section \uppercase\expandafter{\romannumeral4}. The performance gain in terms of the throughput and the delay are analyzed in Section \uppercase\expandafter{\romannumeral5}. In Section \uppercase\expandafter{\romannumeral6}, numerical results are presented. Finally, we give our conclusions in Section \uppercase\expandafter{\romannumeral7}.
\section{System model and Protocol description}\label{sec:system}
In this section, we first model the nodes of the three-tier HetNet as mutually independent PPPs with different densities. Then the  cache-enabled content access protocol is described. Afterwards, the probability of the tier association priority and the state of users are derived.
\subsection{Network Architecture}
Consider a three-tier wireless HetNet consisting of a number of macro BSs, relays and users as illustrated in Fig. \ref{system_model}. The nodes of the $i$-th tier ($i=0, 2, 3$ for the users, relays and BSs, respectively) are deployed based on an independent homogeneous PPP $\psi_i$ with intensity $\lambda_i$ \cite{capacity1,ppp1,ppp2}. Note that in the practical system there are more users than relays or BSs, so we consider $\lambda_0\gg\lambda_2>\lambda_3$ in this paper.  There are $N$ multimedia contents on the multimedia server, where all the contents are assumed to have the same size of $S$ $[$bits$]$. Each of the relays has a limited caching storage with the size of $M_2\times S$ $[$bits$]$, but only a part (e.g., the $0<\alpha<1$ proportion) of the users has caching ability and the corresponding size is $M_1 \times S$ $[$bits$]$, and $M_1\ll M_2\ll N$. According to Poisson processes, the locations of the cache-enabled users are distributed as a thinning homogeneous PPP with density $\lambda_1=\alpha\lambda_0$.

\begin{figure}[t]
\centering
\includegraphics[width=3.5in]{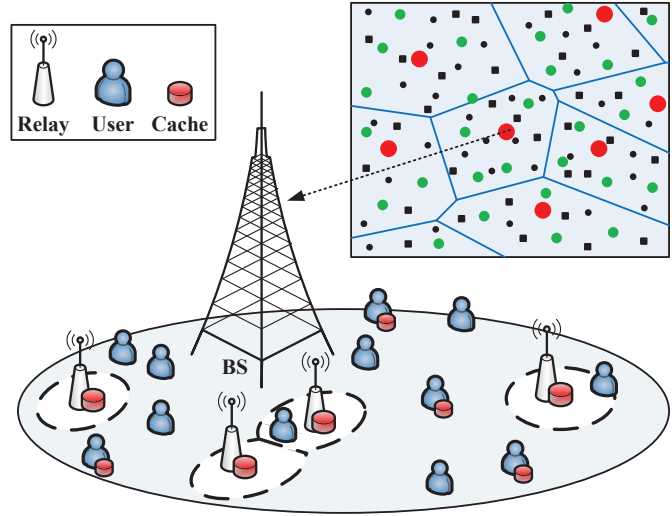}
\caption{Cooperative caching in heterogeneous networks: The plot on the top right side is a snapshot of different nodes deployed with PPPs: BSs (red circle) are overlaid with relays (green circle). Parts of users (black circle) have caching ability, the others (black rectangle) do not. The structure of a typical BS cell is highlighted in the lower plot.}
\label{system_model}
\end{figure}

It has been observed that people are always interested
in the most popular multimedia contents, where only a small portion of the contents are frequently accessed by the majority of users \cite{Zipf}. The higher ranking of a multimedia content, the greater the requested probability. The popularity of the $i$-ranked content can be modeled by the Zipf distribution as follow\cite{Zipf,kongtao,scaling}
\begin{equation}
f_{i}={\frac{1/i^{\gamma}}{\sum_{j=1}^{N}1/j^{\gamma}}},  \label{zipf}
\end{equation}%
where $\gamma\geq 0$ reflects the skew of the content popularity distribution. The larger $\gamma$, the fewer of popular contents accounting for the majority of the requests.
\subsection{Cache-enabled Content Access Protocol}\label{subsec:protocol}
 A high-capacity wired backhual solution can be used for the connection link between BSs and relays, e.g., optical fiber. When the network is at low traffic load, e.g., the traffic load in the nighttime, the most popular contents can be cached at the relays and the cache-enabled users via broadcasting \cite{cache-enable}. All the cache-enabled users store the same copy of the contents until the caching storage is fully occupied, and those cached in different relays are also the same.

When a user requests a multimedia content, it first checks whether the caching storage is available in its local devices. If the requested content is cached in its caching storage, the user can obtain the content immediately; otherwise, the user access the ``closest'' node. Here, we define the node providing the maximum received power as the ``closest'' node of a requesting user. The user's received power is defined as \cite{capacity1,capacity2,SINR},
 \begin{equation}
C_{i}=\nu B_iP_ir_i^{-\beta},
 \end{equation}
where $P_i$ for $i=1,2,3$ is the transmit power of the node in the $i$-th tier. When $i=1$ it means the user gets the content from a cache-enabled user via the local sharing link such as D2D \cite{femto2}\cite{grid}\cite{trade-off}\cite{scaling} considered in this paper. $\beta\geq 2$ denotes the path-loss exponent and $r_i$ is the distance between the requesting user and its closest node of the $i$-th tier. $\nu$ denotes a propagation constant and is normalized as $1$ in this paper. For clarity, the association bias $B_{i}$ of the $i$-th tier is assumed to be 1. Thus, the closest node is $\arg\max_{i}C_{i}$. As a result, there are four content access Cases in the three-tier HetNet as described below.

\textbf{Case~1:} The requesting user does not have caching ability, and it can successfully obtain the requested contents from the closest node (BS, relay or the cache-enabled user).

\textbf{Case~2:} The requesting user has caching ability, but the requested content is not cached in its local caching storage, which means all the cache-enabled users do not store the requested content. Thus, only the closest node in the relay or BS tier can respond to the request.

\textbf{Case~3:} The requesting user does not have caching ability, and its closest node is a cache-enabled user. However, the corresponding cache-enabled user does not cache the requested content due to the limited caching storage. So the requesting user needs to obtain the content from the closest node in other tiers, i.e., the relay or BS.

\textbf{Case~4:} The requesting user has caching ability, and it can obtain the requested content from its local caching storage immediately.

The HetNet without caching is considered as a baseline in this paper. In the baseline,
neither users nor relays  have caching ability. Therefore, BSs could not pre-broadcast the popular contents to the users and relays. And the local sharing links (D2D) among users can not work at this time. If the closest node of a user is a relay, the relay needs to fetch the content from the BS via wired backhaul firstly and then forwards it to the user; If the closest node is a BS, the BS responds the user's request. Similarly, in the caching network, if the content is not cached in the relay, a $\emph{Backhaul-needed}$ ($\emph{BH-needed}$) event happens and the relay needs to fetch the content from the BS via the backhaul firstly; otherwise, the $\emph{Backhaul-free}$ ($\emph{BH-free}$) event happens and the relay can respond the request immediately without the backhaul.
\subsection{The Probability of the Tier Association Priority}
As described above, the locations of users, BSs and relays are modeled as mutually independent PPPs. Therefore, the probability that there are $n$ nodes
in area $A$ with radius of $r$ is given by
\begin{equation}
\mathbb{P}\left(\begin{array}{c|}
                     n\,\text{in}\,\psi_i
                   \end{array}~A=\pi r^2\right)=e^{-\pi r^2\lambda_{i}}\frac{(\pi r^2\lambda_{i})^{n}}{n!}\text{,}
\end{equation}
where $n=0,1,2,...$ and $i=0,1,2,3$.
Without loss of generality, according to Slivnyak's theorem \cite{origin}, we conduct analysis on assumption that there is a typical user with or without caching ability at the origin of the Euclidean area, and it is regarded as the reference user.

We first analyze the scenario where the reference user does not have caching ability with the probability of $1-\alpha$. So the probability that the distance between the reference user and its closest cache-enabled user is larger than $r_{1}$ is
\begin{equation}
\mathbb{P}\left(y\geq r_{1}\right)=\mathbb{P}\left(0\,\text{in}\,\psi_1 | r_1\right)=e^{-\pi\lambda_{1}r_{1}^{2}}.
\end{equation}
Therefore, the probability density function (PDF) of the distance from the reference user to the closest cache-enabled user is given by
\begin{equation}
f_{R_1}(r_1)=\frac{\partial\left(1-\mathbb{P}\left(y\geq r_{1}\right)\right)}{\partial r_1}=2\pi\lambda_{1}r_{1}e^{-\pi\lambda_1r_{1}^{2}}.
\end{equation}
Similarly, the PDF of the distance from the reference user to its closest relay and BS are
\begin{equation}
f_{R_i}(r_i)=2\pi\lambda_{i}r_{i}e^{-\pi\lambda_ir_{i}^{2}}, i=2,3,
\end{equation}
respectively. As a result, the joint PDF can be given by
\begin{equation}\label{r1r2r3}
f_{R_1,R_2,R_3}(r_1,r_2,r_3)=\left(\prod_{i=1}^{3}2\pi\lambda_{i}r_{i}\right)e^{-\pi\sum\limits_{i=1}^{3}\lambda_i r_{i}^{2}}.
\end{equation}
To derive the main conclusions in the following, we first consider the general $K$-tier HetNet with PPPs with parameters $\lambda_i$ and $P_i,i=1,2,...,K$. Denote $C_{t_i},i=1,2,...,K,$ as the maximum received-power from
the $t_i$-th tier, where $t_i\in\{1,2,...,K\}$ means that the value of the maximum received-power from the $t_i$-th tier is ranked $i$-th. We thus have the following proposition.
\begin{proposition}\label{ccc}
The probability of $C_{t_1}\!>\!C_{t_2}\!>\!...\!>\!C_{t_K}$ is
\begin{align}\
&\mathbb{P}(C_{t_1}\!>\!C_{t_2}\!>\!...\!>\!C_{t_K})=\prod_{n=1}^{K-1}\left[\sum_{m=n}^K\frac{\lambda_{t_m}}{\lambda_{t_n}}\left(\frac{P_{t_m}}{P_{t_n}}\right)^\frac{2}{\beta}\right]^{-1}.
\end{align}
\end{proposition}
\begin{proof}
See Appendix A.
\end{proof}

For the three-tier HetNet of this paper, the probability of $C_{i}>C_{j}>C_{k}, i\neq j \neq k\in\{1,2,3\}$ is then given by
\begin{align}
&\mathrm{\mathbb{P}(C_{i}>C_{j}>C_{k})}={\left[1+\frac{\lambda_k}{\lambda_j}\left(\frac{P_k}{P_j}\right)^\frac{2}{\beta}\right]^{-1}\left[\sum\limits_{j=1}^3\frac{\lambda_j}{\lambda_i}\left(\frac{P_j}{P_i}\right)^{\frac{2}{\beta}}\right]^{-1}}.\label{probability}
\end{align}	
We observe from (\ref{probability}) that the reference user without caching ability prefers to obtain the content from the $i$-th, $j$-th, and $k$-th tier in turn with probability of $\mathrm{\mathbb{P}(C_{i}>C_{j}>C_{k})}$.  Proposition \ref{ccc} can be further extended to the following lemma.
\begin{lemma}\label{ddd}
The probability that the reference user without caching ability prefers to associate with the $i$-th tier at first in the $K$-tier HetNet is
\begin{align}
\mathcal{G}_{K,i}\triangleq \mathbb{P}(C_{i}\!>\!\max_{\forall n\neq i}C_{n})=\left[\sum_{m=1}^K\frac{\lambda_{m}}{\lambda_{i}}\left(\frac{P_{m}}{P_{i}}\right)^\frac{2}{\beta}\right]^{-1}.
\end{align}
\end{lemma}
\begin{proof}
See Appendix B.
\end{proof}

So in the three-tier HetNet, the probability that the reference user without caching prefers to get the content from the $i$-th tier at first is
\begin{equation}\label{prefer}
    \mathcal{G}_{3,i}=\left[{\sum\limits_{j=1}^3\frac{\lambda_j}{\lambda_i}\left(\frac{P_j}{P_i}\right)^{\frac{2}{\beta}}}\right]^{-1}.
\end{equation}

Likewise, as to the scenario where the reference user is cache-enabled, the probability of $C_{i}>C_{j}, i\neq j\in\{2,3\}$ is
\begin{align}
 \mathrm{\mathbb{P}(C_{i}>C_{j})}\!=\left[{\sum\limits_{j=2}^3\frac{\lambda_j}{\lambda_i}\left(\frac{P_j}{P_i}\right)^{\frac{2}{\beta}}}\right]^{-1}.\label{probability_2}
\end{align}	
Equation (\ref{probability_2}) means that the reference user with caching ability prefers to obtain the content from the $i$-th, $j$-th tier in turn
with probability of $\mathrm{\mathbb{P}(C_{i}>C_{j})}$ when the requested content has not been cached in the local caching. $\mathrm{\mathbb{P}(C_{i}>C_{j}>C_{k})}$ and $\mathrm{\mathbb{P}(C_{i}>C_{j})}$ will be denoted as $\mathrm{\mathbb{P}_{i,j,k}}$ and $\mathrm{\mathbb{P}_{i,j}}$ respectively for convenience in the following. According to  (\ref{prefer}) and (\ref{probability_2}), we find that the tier association priorities are different when the user is cache-enabled or not. Users prefer to connect to the tier with higher transmit power and node density.
\subsection{The Density of the Active D2D Transmitters}\label{sec:system_4}
In the subsection above, we have analyzed the tier association priority merely based on the geographical locations, where the detailed impacts of the limited caching space and the content popularity are not considered. Define $C\in\{\text{Case 1, Case 2, Case 3, Case 4}\}$ as the Case the user may be active in. Let $T\in\{\text{Tier 1, Tier 2, Tier 3, Local}\}$ be the node where the user can obtain contents. Let $W\in\{\emph{BH-needed}, \emph{BH-free}\}$ describe whether the backhaul is needed for a user to access the content successfully. We only consider the wired backhaul between the BS and the relay, the impacts of the backhaul from the multimedia server to the BS are out of the scope of this paper. Denote $\mathbf{\chi}=(C,T,W)$ as the state of the user. Probabilities of different $\mathbf{\chi}$ are listed in Table \ref{table:pushing}, where we rewrite $\sum_{i=a}^bf_i$ as $F(a,b)$ for simplification. Assign the value in the $i$-th $i\in\{2,3,...,9\}$ row $j$-th $j\in\{3,4,...,6\}$  column of Table \ref{table:pushing} to the element $D_{i-1,j-2}$ of a matrix $\mathbf{D}_{8\times4}$.
\begin{table}
\caption{Probabilities of that the user is active in different states.}
\vspace{-0.8cm}
\label{table:pushing}
\begin{center}
\begin{tabular}{|c|c|c|c|c|c|}
\hline
\multicolumn{2}{|c|}{$\mathbb{P}[\chi=(C,T,W)]$} &\multicolumn{1}{|c|}{\textbf{Tier 1 (D2D)}} &\multicolumn{1}{|c|}{\textbf{Tier 2 (Relay)}}&\multicolumn{1}{|c|}{\textbf{Tier 3 (BS)}}&\!\!\!\textbf{Local}\!\!\! \\
\hline
\multirow{2}{*}{\textbf{\!\!Case 1\!\!}}&\!\!$\emph{{BH-free}}$\!\!&$\!\!\mathcal{G}_{3,1}(1\!-\!\alpha)F(1,M_1)\!\!$&$\mathcal{G}_{3,2}(1\!-\!\alpha)F(1,M_2)$&$\mathcal{G}_{3,3}(1\!-\!\alpha)$&$0$\\
\cline{2-6}
&\!\!$\emph{{BH-needed}}$\!\!&$0$&$\mathcal{G}_{3,2}(1\!-\!\alpha)F(M_2\!+\!1,N)$&0&$0$\\
\hline
\multirow{2}{*}{\textbf{\!\!Case 2\!\!}}&$\emph{{BH-free}}$&$0$&$\mathbb{P}_{2,3}\alpha F(M_1\!+\!1,M_2)$&$\mathbb{P}_{3,2}\alpha F(M_1\!+\!1,N)$&$0$\\
\cline{2-6}
&\!\!$\emph{{BH-needed}}$\!\!&$0$&$\mathbb{P}_{2,3}\alpha F(M_2\!+\!1,N)$&0&$0$\\
\hline
\multirow{2}{*}{\textbf{\!\!Case 3\!\!}}&$\emph{{BH-free}}$&$0$&\!\!$\mathbb{P}_{1,2,3}(1\!-\!\alpha) F(M_1\!+\!1,M_2)$\!\!&\!\!$\mathbb{P}_{1,3,2}(1\!-\!\alpha)F(M_1\!+\!1,N)\!\!$&$0$\\
\cline{2-6}
&\!\!$\emph{{BH-needed}}$\!\!&$0$&$\mathbb{P}_{1,2,3}(1\!-\!\alpha)F(M_2\!+\!1,N)$&0&$0$\\
\hline
\multirow{2}{*}{\textbf{\!\!Case 4\!\!}}&$\emph{{BH-free}}$&$0$&$0$&0&$\!\!\alpha F(1,M_1)\!\!$\\
\cline{2-6}
&\!\!$\emph{{BH-needed}}$\!\!&$0$&0&0&$0$\\
\hline
\end{tabular}
\end{center}
  \vspace{-0.7cm}
\end{table}

Based on Table \ref{table:pushing}, the probability that a user obtains the content successfully via the D2D link is $\mathcal{G}_{3,1}(1\!-\!\alpha)F(1,M_1)$, i.e., $D_{1,1}$. So the density of  users to be served by D2D transmitters (TXs) is $\lambda_0\mathcal{G}_{3,1}(1\!-\!\alpha)F(1,M_1)$. However, the density of the cache-enabled user is $\lambda_0\alpha$, which is the maximum density of D2D TXs. Define $\lambda'_1$ as the density of the actually active D2D TXs. If a small fraction of users have caching ability, in the coverage of a cache-enabled user, there is at least one user to be responded via the D2D link, i.e., the density $\lambda'_1$ is $\alpha\lambda_0$. At this time the number of D2D links are limited by the number of cache-enabled users. All of cache-enabled users should be active as D2D TXs to satisfy the demand for the D2D link. However, if most of users are cache-enabled, some cache-enabled users may not cover any user in the corresponding coverage. The density of cache-enabled users active as  D2D TXs is $\lambda'_1=(1-\alpha)\lambda_0\mathcal{G}_{3,1}F(1,M_1)$, which is smaller than $\alpha\lambda_0$. At this time the number of D2D links are limited by the number of the users without caching ability, and not all of cache-enabled users are active as D2D TXs. Thus, the node density of the active D2D TXs can be given by
\begin{equation}\label{lambda1}
\lambda'_1=\text{min}\left\{\alpha\lambda_0,(1-\alpha)\lambda_0\mathcal{G}_{3,1}F(1,M_1)\right\}\text{.}
\end{equation}

We define ${\alpha}^*$ as the critical point deciding whether all of cache-enabled users need to be active as D2D TXs. Let $\alpha\lambda_0=(1-\alpha)\lambda_0\mathcal{G}_{3,1}\sum_{i=1}^{M_1}f_i$ we get the critical point,
\begin{equation}\label{alpha}
\alpha^*=\text{max}\left\{0, \left[{F(1,M_1)-h}\right]\left[{1+F(1,M_1)}\right]^{-1}\right\},
\end{equation}
where $h=\sum_{j=2}^3\frac{\lambda_j}{\lambda_0}(\frac{P_j}{P_1})^{\frac{2}{\beta}}$. From (\ref{alpha}) we observe that whether all the cache-enabled user need to be active as D2D TXs is jointly decided by the user caching ability ($M_1$), the content popularity ($\gamma$), the transmit power ($P_i$), the node density ($\lambda_i$) and the path-loss exponent ($\beta$). $\alpha^*$ increases with the increase of $M_1,\gamma$ and $\beta$. Then equation (\ref{lambda1}) can be rewritten as
\begin{equation}\label{ui_zst}
\lambda'_1 = \left\{
\begin{array}{ll}
\alpha\lambda_0,&\alpha<\alpha^*;\\
(1-\alpha)\lambda_0\mathcal{G}_{3,1}F(1,M_1),&\alpha\geq\alpha^*.
\end{array}
 \right.
\end{equation}

Next, we introduce another variable $\widehat\alpha$ which decides the maximum density of the D2D TXs in the network with various $\alpha$. Based on the first derivative of $(1-\alpha)\lambda_0\mathcal{G}_{3,1}F(1,M_1)$ with respect to $\alpha$ we can get $\widehat{\alpha}=\sqrt{h^2+h}-h$. The density of the D2D TXs increases with $\alpha$ in the region $[0,\widehat{\alpha}]$ and starts to decrease from $\widehat{\alpha}$. It implies at most $\widehat{\alpha}$ D2D links can be set up in unit area.

For the other two tiers, as considered above that  all the relays and BSs are fully loaded and active when $\lambda_0\gg\lambda_2>\lambda_3$, the actually active node density of the BSs and relays equal to the corresponding node density, i.e., $\lambda_{i}^{'}=\lambda_{i}$ for $i=2,3$. Therefore, the nodes of the actually active D2D TXs, relays and BSs are scattered according to mutually  independent homogeneous PPPs $\Phi_i, i=1,2,3$ with the density $\lambda'_i$, respectively.
\section{The Average Ergodic Rate}\label{sec:rate}
The average ergodic rate in the downlink is analyzed in this section. Specifically, the communication link between the relay/user and the requesting user is assumed to share the same frequency with that from the BS to the users, yielding the interference. There exist two types of interferences, namely, the inter-tier and the intra-tier interference. The full load state of the BS and relay is considered and user requests arriving at the same service node are responded one after the other in a round-robin manner  \cite{capacity1}. We shall note that the rate analyzed in this Section refers to that over the air, and the effect of the backhaul will be considered in Section \ref{sec:throughput}.

Therefore, the signal-to-interference-plus-noise ratio (SINR) of the reference user associated with the node in the $i$-th tier is
\begin{align}\label{sinr}
    \text{SINR}_i(x)=\frac{P_ig_{i,0}x^{-\beta}}{\sum\limits_{j=1}^3\sum\limits_{k\in\Phi_j\setminus B_{i,0}}P_jh_{jk}|Y_{jk}|^{-\beta}+\sigma^2}\triangleq\frac{P_ig_{i,0}x^{-\beta}}{\sum\limits_{j=1}^3I_j+\sigma^2}\triangleq\frac{P_ig_{i,0}x^{-\beta}}{I_r},
\end{align}
where $\sigma^2$ denotes the power of the additive noise, $x$ is the distance between the reference user and its serving node. $g_{k,0}$ and $h_{jk}$ denote the channel power gain. Here, we consider Rayleigh fading
channels with average unit power, yielding $g_{k,0}\sim \exp(1),h_{jk}\sim \exp(1)$. $|Y_{jk}|$ is the distance between the reference user and its interfering nodes $k$ in the $j$-th tier. $I_{j}$ denotes the cumulative interference from the $j$-th tier. We define the average ergodic rate $\mathcal{U}_i,i=1,2,3$ of the reference user when it communicates with the $i$-th tier as \cite{capacity1,ppp2,ppp3,ppp1,SINR},
\begin{equation}\label{rate2}
    \mathcal{U}_i\triangleq\mathbb{E}_x[\mathbb{E}_{\text{SINR}_i}[\text{ln}(1+\text{SINR}_i(x))]].
\end{equation}
Here, the unit of the average rate is nats/s/Hz ($1 ~\text{nat}=1.443~\text{bits}$) to simplify the analysis.
The average is taken over both the channel fading distribution and the spatial PPP. The ergodic rate is first averaged on condition that the reference user is at a distance $x$ from its serving node in the $i$-th tier. Then the rate is averaged via calculating the expectation over the distance $x$. The metric means the average ergodic rate of a randomly chosen user associated to the $i$-th tier.
\subsection{The Average Ergodic Rate in  Case 1}
Denote $X_i$ as the distance between the reference user and its serving node of tier $i$. Based on the proof in \cite{SINR}, we can obtain the PDF of $X_i$ as follow,
\begin{equation}\label{fx}
   f_{X_i}(x)=\frac{2\pi \lambda_i}{\mathcal{G}_{3,i}}xe^{-\pi \sum\limits_{j=1}^3\lambda_j \left(\frac{P_j}{P_i}\right)^{\frac{2}{\beta}}x^2}.
\end{equation}
Then we have the following theorem.
\begin{theorem}\label{rateltheorem}
The average ergodic rate   of the reference user associated with the $i$-th tier ($i=1,2,3$) in Case 1 is
\begin{align}\label{ui_z}
\mathcal{U}_{1,i}=\frac{2\pi \lambda_i}{\mathcal{G}_{3,i}}\int_0^\infty\!\!\!\int_0^\infty x \text{exp}\left\{-x^\beta P_i^{-1}(e^t-1)\sigma^2\!-\!\frac{\pi \lambda_ix^2}{\mathcal{G}_{3,i}}\left[1\!+\!\frac{\lambda_1+(\lambda'_1\!-\!\lambda_1)\mathcal{G}_{3,1}}{\lambda_1\mathcal{Z}_1^{-1}(e^t\!-\!1)}\right]\right\}\mathrm{d}{t}\mathrm{d}{x}.
\end{align}
\end{theorem}
\begin{proof}
See Appendix \ref{sec:prooftheroem1}.
\end{proof}

Since the node densities are typically quite high in the HetNet, the background noise is far smaller than the interference power. The interference is dominant and the noise can often be neglected,
i.e. $(\sigma^2\rightarrow0)$, then the rate is further simplified to
\begin{align}\label{ui_zs1}
\mathcal{U}_{1,i}=\mathlarger{\int}_0^\infty\frac{1}{1+\frac{\lambda_1+(\lambda'_1-\lambda_1)\mathcal{G}_{3,1}}{\lambda_1\mathcal{Z}_1^{-1}(e^t\!-\!1)}}\mathrm{d}{t}.
\end{align}
According to (\ref{ui_zst}), equation (\ref{ui_zs1}) can be further rewritten as
\begin{equation}\label{ui_zs}
\mathcal{U}_{1,i} = \left\{
\begin{array}{ll}
\mathlarger{\mathlarger{\int}}_0^\infty\frac{1}{1+\mathcal{Z}_1(e^t\!-\!1)}\mathrm{d}{t},&\alpha<\alpha^*;\\
{\mathlarger{\mathlarger{\int}}}_0^\infty\left[1+\left(1+\frac{1-\alpha}{\alpha}\mathcal{G}_{3,1}^2\sum\limits_{i=1}^{M_1}f_i-\mathcal{G}_{3,1}\right)\mathcal{Z}_1(e^t\!-\!1)\right]^{-1}\mathrm{d}{t},&\alpha\geq\alpha^*.
\end{array}
 \right.
\end{equation}

Equations (\ref{ui_zs1}) and (\ref{ui_zs}) reveal that when the interference is dominant, the average ergodic rate in Case 1 is independent on which tier the user connects to. Moreover, from (\ref{ui_zs}) we know that, when the interference is dominant and
the fraction of the cache-enabled users is small~$(\text{i.e.},\alpha<\alpha^*)$, the average ergodic  rate keeps constant. On one hand, the rate  keeps constant when $\alpha$ varies in the region of $[0, \alpha^*]$ such that all of the cache-enabled users should be active as D2D TXs. Higher density of D2D TXs gets the content closer to users while adding additional interference caused by the increase of the D2D pairs. On the other hand, the rate keeps constant independently on system parameters such as the transmit power $P_i$ and node density. This means that raising the transmit power or service node densities increases the desired signal power and the interference by the same amount, and they offset each other. However, the parameters affect the number of simultaneously active nodes in unit area. As an example, with larger $\alpha$ and caching ability $M_1$, more users can get contents via the D2D link or immediately from their local caching,  yielding the change of the sum rate of the cache-enable network. However, when $\alpha>\alpha^*$, the average ergodic  rate increases with the increase of $\alpha$ based on (\ref{ui_zs}). It is because the distance between the user and the D2D TX is reduced with larger number of cache-enabled users, but not all the cache-enabled users need to be active as the D2D TXs at this time, breaking the balance between the desired signal power and the interference.
\subsection{The Average Ergodic Rate in  Case 2}
Similar to (\ref{fx}), the PDF of the distance between the reference user and its serving node of tier $i$ in Case 2 is
\begin{equation}\label{fx2}
   \widehat{f}_{X_i}(x)=\frac{2\pi \lambda_i}{\mathbb{P}_{i,j}}xe^{-\pi \sum\limits_{j=2}^3\lambda_j \left(\frac{P_j}{P_i}\right)^{\frac{2}{\beta}}x^2}, i, j\in \{2,3\}, i\neq j.
\end{equation}
We then calculate the average ergodic rate for Case 2 as follow.
\begin{theorem}\label{rate2theorem}
The average ergodic rate   of the reference user associated with the $i$-th tier ($i=2,3$) in Case 2 is
\begin{align}\label{ui_z_2}
&\mathcal{U}_{2,i}{\setlength\arraycolsep{0.5pt}=}\frac{2\pi  \lambda_i}{\mathbb{P}_{i,j}}\!\!\int_0^\infty\!\!\!\int_0^\infty\!\!x\text{exp}\left\{{\setlength\arraycolsep{0.5pt}-}x^\beta P_i^{{\setlength\arraycolsep{0.5pt}-}1}(e^t{\setlength\arraycolsep{0.5pt}-}1)\sigma^2{\setlength\arraycolsep{0.5pt}-}\frac{\pi \lambda_ix^2}{\mathbb{P}_{i,j}} \left[1{\setlength\arraycolsep{0.5pt}+}\mathcal{Z}_1(e^t{\setlength\arraycolsep{0.5pt}-}1){\setlength\arraycolsep{0.5pt}+}\frac{\lambda'_1}{\lambda_1}\frac{\mathcal{G}_{3,1}}{1{\setlength\arraycolsep{0.5pt}-}\mathcal{G}_{3,1}}\mathcal{Z}_2(a)\right]\right\}\mathrm{d}{t}\mathrm{d}{x}.
\end{align}
\end{theorem}
\begin{proof}
See Appendix \ref{sec:prooftheroem2}.
\end{proof}

When the interference is dominant, i.e., $\sigma^2\rightarrow0$, we obtain
\begin{align}\label{ui_zsL}
\mathcal{U}_{2,i}=\int_0^\infty\!\!\!\!\frac{1}{1+\mathcal{Z}_1(e^t\!-\!1)+\frac{\lambda'_1}{\lambda_1}\frac{\mathcal{G}_{3,1}}{1-\mathcal{G}_{3,1}}\mathcal{Z}_2(a)}\mathrm{d}{t}.
\end{align}
According to (\ref{ui_zst}), (\ref{ui_zsL}) can be rewritten as
\begin{equation}\label{ui_zsLL}
\mathcal{U}_{2,i} = \left\{
\begin{array}{ll}
\mathlarger{\int}_0^\infty\frac{1}{1+\mathcal{Z}_1(e^t\!-\!1)+\frac{\mathcal{G}_{3,1}}{1-\mathcal{G}_{3,1}}\mathcal{Z}_2(a)}\mathrm{d}{t},&\alpha<\alpha^*;\\
\mathlarger{\int}_0^\infty\left[1+\mathcal{Z}_1(e^t\!-\!1)+\frac{\left(1-\alpha\right)\mathcal{G}_{3,1}^2}{\alpha\left(1-\mathcal{G}_{3,1}\right)}\sum\limits_{i=1}^{M_1}f_i\mathcal{Z}_2(a)\right]^{-1}\mathrm{d}{t},&\alpha\geq\alpha^*.
\end{array}
 \right.
\end{equation}

Compared (\ref{ui_zsLL}) with (\ref{ui_zs}), we find that the  average ergodic rate of Case 2 is smaller than that of Case 1. D2D TXs  bring out additionally unnecessary interference to the users in Case 2, decreasing the rate. According to $(\ref{ui_zsLL})$, when $\alpha<\alpha^*$, the average ergodic rate in Case 2  decreases with the increase of $\alpha$  because $\mathcal{G}_{3,1}$ increases with the $\alpha$. Furthermore, set the first derivative of $\mathcal{U}_{2,i}$ with respect to $\alpha$ to zero, we can get a critical point $\widehat{\alpha}$, which exactly is the point getting the maximum number of active D2D TXs explained in subsection \ref{sec:system_4}. The number of active D2D TXs increases monotonically with the increase of $\alpha$  when $\alpha<\widehat{\alpha}$; otherwise, it decreases monotonically when $\alpha\geq\widehat{\alpha}$. More active D2D TXs lead to more unnecessary interference to the users of Case 2. Therefore, the rate  continues decreasing  with the increase of $\alpha$ in the region $[\alpha^*,\widehat{\alpha}]$ but it starts to increase from $\widehat{\alpha}$. On the whole, the rate in Case 2 decreases in the region $[0,\widehat{\alpha}]$ and then increases from $\widehat{\alpha}$. Any  of the network parameters such as the node density ($\lambda_i$), the content popularity ($\gamma$), the transmit power ($P_i$), the path-loss parameter ($\beta$) and the caching ability ($M_1$)  can affect the trend of the rate in Case 2.
\subsection{The Average Ergodic Rate in  Case 3}
Based on the definition of Case 3, we have $C_1>C_j>C_k, (j,k)\in\{(2,3),(3,2)\}$. As described above, $X_{1}$ is the distance between the reference user and its closest cache-enabled user. Let $Y_j$  be the distance between the reference user and its closest node in the $j$-th tier for $j=2,3$. Then the joint PDF of $x,y$ in Case 3 is
\begin{align}\label{fxy}
&f_{X_1,Y_j}(x,\!y)=
{\frac{4\pi^2\lambda_1\lambda_jxy}{\mathbb{P}_{1,j,k}}}{\text{exp}\left\{{-\pi\lambda_1x^2-\pi\lambda_jy^2\left[1+\frac{\lambda_k}{\lambda_j}(\frac{P_k}{P_j})^{\frac{2}{\beta}}\right]}\right\}}, ~\text{if}~ y\!>\!(\frac{P_j}{P_1})^{\frac{1}{\beta}}x.
\end{align}
If $y\!\leq\!(\frac{P_j}{P_1})^{\frac{1}{\beta}}x$, $f_{X_1,Y_j}(x,\!y)=0$. The proof is derived in Appendix \ref{sec:fxy}. As a result, we obtain the following theorem,
\begin{theorem}\label{rate3theorem}
When the interference is dominant, the average ergodic rate   of the reference user associated with the $j$-th tier ($j=2,3$) in Case 3 is
\begin{align}\label{rate3}
\!\!\!\mathcal{U}_{3,j}\!=\!\!\!\int_0^\infty\!\!\!\!\int_0^{1}\!\!\!\frac{2x(1\!-\!\mathcal{G}_{3,1})^{-1}}{\!\left\{1\!+\!\mathcal{Z}
_1(e^t\!-\!1)\!+\!\frac{\mathcal{G}_{3,1}x^2}{1\!-\!\mathcal{G}_{3,1}}\left[1\!+\!\frac{\lambda'_1}{\lambda_1}\mathcal{Z}_3(e^t\!-\!1)\right]\!\right\}^2}\mathrm{d}{x}\mathrm{d}{t}.
\end{align}
\end{theorem}
\begin{proof}
See Appendix \ref{sec:prooftheroem3}.
\end{proof}

We can observe that when the interference is dominant, the average ergodic rate of the reference user in Case 3 is also independent on which tier the user connects to. Furthermore, denote $\mathcal{U}_l$ as the average ergodic rate of the reference user in Case 4. $\mathcal{U}_l$ shall be considered as a extremely fast speed with which the user can read out the contents from its local caching disk immediately. The higher the content popularity ($\gamma$) and caching ability ($M_1$) become, the higher probability there is for users to be active in Case 4.
\section{The Outage Probability}
Besides the average ergodic rate elaborated in the previous section, we will derive another important performance metric, i.e., the outage probability in this section. The outage probability can be defined as the probability that the instantaneous SINR of a randomly located user
is less than a threshold $\tau$. Let $\mathcal{P}_i$ be the average outage probability of the reference user associated with the $i$-th tier, which can be expressed as \cite{capacity1,ppp2,ppp3,ppp1,SINR,cache-enable},
\begin{equation}\label{outage1}
\mathcal{P}_i\triangleq\mathbb{E}[\mathbb{P}[\text{SINR}_i(x)\leq \tau]].
\end{equation}
The metric can be equivalently interpreted as the average fraction of the cell area where the receiving SINR is smaller than a specific threshold. It is also exactly the cumulative distribution function (CDF) of the SINR over the entire network.
The outage probabilities of the different Cases are analyzed in the following. As to Case 1, we have the following theorem.
\begin{theorem}\label{outageltheorem}
The average outage probability of the user connected to the $i$-th tier ($i=1,2,3$) in Case 1 is
\begin{align}\label{oi_z}
\!\!\mathcal{P}_{1,i}\!=\!1\!-\!\frac{2\pi \lambda_i}{\mathcal{G}_{3,i}}\!\!\int_0^\infty\!\!xe^{-\frac{x^\beta \sigma^2 \tau}{P_i}-\frac{\pi \lambda_ix^2}{\mathcal{G}_{3,i}}\left[1+\frac{\lambda_1+(\lambda'_1-\lambda_1)\mathcal{G}_{3,1}}{\lambda_1\mathcal{Z}_1^{-1}(\tau)}\right]}\mathrm{d}{x}.
\end{align}
\end{theorem}
\begin{proof}
See Appendix \ref{sec:prooftheroem4}.
\end{proof}

For the special scenario where the interference is dominant, i.e., $\sigma^2\rightarrow0$, we have
\begin{align}
\mathcal{P}_{1,i}=1-\left[{1+\frac{\lambda_1+(\lambda'_1-\lambda_1)\mathcal{G}_{3,1}}{\lambda_1\mathcal{Z}_1^{-1}(\tau)}}\right]^{-1}.
\end{align}
\begin{theorem}\label{outage2theorem}
The average outage probability of the user connected to the $i$-th tier ($i=2,3$) in Case 2 is
\begin{align}\label{oi_z_2}
\mathcal{P}_{2,i}=&1-\!\frac{2\pi \lambda_i}{\mathbb{P}_{i,j}}\!\!\int_0^\infty\!\!x\text{exp}\left\{\!-x^\beta P_i^{-1}\tau\sigma^2\!-\!\frac{\pi \lambda_ix^2}{\mathbb{P}_{i,j}} \left[1+\mathcal{Z}_1(\tau)+\frac{\lambda'_1}{\lambda_1}\frac{\mathcal{G}_{3,1}}{1-\mathcal{G}_{3,1}}\mathcal{Z}_2(a)\right]\right\}\mathrm{d}{x},
\end{align}
where $\mathcal{Z}_2(a)=\tau^{\frac{2}{\beta}}\frac{2a^{\frac{2-\beta}{2}}}{\beta-2}{}_2F_1[1,1-\frac{2}{\beta};2-\frac{2}{\beta};-a^{\frac{-\beta}{2}}]$ and $a$ is as small as $0$.
\end{theorem}
\begin{proof}
See Appendix \ref{sec:prooftheroem5}.
\end{proof}

When the interference is dominant, we get
\begin{align}
\mathcal{P}_{2,i}=1-\frac{1}{1+\mathcal{Z}_1(\tau)+\frac{\lambda'_1}{\lambda_1}\frac{\mathcal{G}_{3,1}}{1-\mathcal{G}_{3,1}}\mathcal{Z}_2(a)}.
\end{align}

\begin{theorem}\label{outage3theorem}
When the interference is dominant, the average outage probability of the user connected to the $j$-th tier ($j=2,3$) in Case 3 is
\begin{align}\label{oi_z_3}
&\mathcal{P}_{3,j}\!=\!1\!-\!\!\!\int_0^{1}\!\!\!\frac{2x(1\!-\!\mathcal{G}_{3,1})^{-1}}{\!\left\{1\!+\!\mathcal{Z}
_1(\tau)\!+\!\frac{\mathcal{G}_{3,1}x^2}{1\!-\!\mathcal{G}_{3,1}}\left[1\!+\!\frac{\lambda'_1}{\lambda_1}\mathcal{Z}_3(\tau)\right]\!\right\}^2}\mathrm{d}{x},
\end{align}
where $\mathcal{Z}_3(\tau)=\frac{2\tau}{\beta-2}x^{-\beta}{}_2F_1[1,1-\frac{2}{\beta};2-\frac{2}{\beta};-\tau x^{-\beta}]$.
\end{theorem}
\begin{proof}
See Appendix \ref{sec:prooftheroem6}.
\end{proof}

From Theorems \ref{outageltheorem}, \ref{outage2theorem} and \ref{outage3theorem} we see that, similar with the average egodic rate, the average outage probability is not affected by which tier the user connects to when the interference is dominant. As to users in Case 1, when the fraction of cache-enabled users is small, the interference and the desired signal power change by the same amount with the change of the transmit power or the node densities. However, the unnecessary interference triggered by D2D TXs depraves the outage probability in Case 2 and Case 3. Moreover, let $\mathcal{P}_l$ be the outage probability of the user in Case 4, which is as small as 0 because of the immediate reading from the local caching disk.
\section{The Throughput and The Delay}\label{sec:throughput}
We have analyzed the performance metrics from the perspective of a single user. Based on the analysis results of previous sections, the throughput of the entire network will be derived in this section. The delay and the critical condition for the network to keep steady will be elaborated.

We conduct analysis in a typical BS cell. According to the PPP model for the node locations, the average number of users in a typical BS area is $\frac{\lambda_0}{\lambda_3}$ \cite{SINR}. We now introduce the traffic dynamics of request arrivals and departures. Requests of $\frac{\lambda_0}{\lambda_3}$ users are considered as a unified event and modeled as a Poisson process with parameter $\varsigma$ $[$requests/s$]$, i.e., the request interarrival times are exponentially distributed random variables with mean $\frac{1}{\varsigma}$ seconds \cite{kongtao}\cite{Queue}. It implies that requests of a single user is a Poisson process with parameter $\frac{\varsigma\lambda_3}{\lambda_0}$ $[$requests/s$]$. The arriving requests require to access some sets of contents. Volumes of the sets are independent exponentially distributed random variables with mean $\frac{1}{\varrho}$ $[$contents/request$]$,  and the request interarrivals and request sets are independent \cite{kongtao}\cite{Queue}. We define $\varsigma$ as the total request arrival rate and $\sigma=\frac{\varsigma S}{\varrho}$ as the total traffic demand (in $[$bits/s$]$) in the typical BS cell.

Based on Table \ref{table:pushing}, each element of matrix $\mathbf{D}$ represents the probability that a randomly chosen user is active in the state of $\mathbf{\chi}=(C,T,W)$. Therefore, $(i,j)$ for $i=1,2,...,8$ and $j=1,2,...,4$ can represent the state of the user with the mapping $g:(C,T,W)\rightarrow(i,j)$. Then the density of users in the state of $(i,j)$ is $\lambda_{i,j}=\lambda_0D_{i,j}$. Corresponding to the consideration in Section \ref{sec:rate} that service nodes are in the full load state, in this section, BSs, relays and D2D TXs without user being served are assumed to make dummy transmissions which bring interference to others as well \cite{dummy}. Consider $w$ Hz bandwidth are shared among different tiers. Let element $A_{i,j}$  of matrix $\mathbf{A}_{8\times4}$ denote the average ergodic rate of the user in the state of $(i,j)$. $\mathbf{A}$ is generated by
\begin{equation}\label{ratee}
\left\{
\begin{array}{ll}
A_{2m-1,j}=\eta w~\mathcal{U}_{m,j}\mathbf{1}(D_{2m-1,j}\neq0),&\text{for}~m=1,2,...,4;\\
A_{2m,j}=\eta w f(\mathcal{U}_{m,j})\mathbf{1}(D_{2m,j}\neq0),&\text{for}~m=1,2,...,4,
\end{array}
 \right.
\end{equation}
where $\mathbf{1}(\cdot)$ is the indicator function and $\eta=1.443$ is the conversion factor between $[$nats$]$ and $[$bits$]$. $\mathcal{U}_{m,j}$ is the average ergodic rate  analyzed in Section \ref{sec:rate} and $\mathcal{U}_{4,4}=\mathcal{U}_{l}$. We consider $\mathcal{U}_{2,1}$, $\mathcal{U}_{3,1}$ and $\mathcal{U}_{4,j}$ (for $j=1,2,3$) as $0$ just like we define the matrix $\mathbf{D}$ even though no user is active in these states, and these virtual variables are defined to simplify the description. Due to the delay caused by the additionally wired transmission process and the limited backhaul, we assume the users can get the content with the service rate of $f(\mathcal{U}_{m,j})$ when the backhaul is needed, which is a function of $\mathcal{U}_{m,j}$ and is smaller than $\mathcal{U}_{m,j}$.

In the coverage of a D2D TX (a relay, a BS, a cache-enabled user itself), the average number of users who are in the state of $(i,j)$ is $n_{i,j}=\frac{\lambda_{i,j}}{\lambda'_j}$ \cite{SINR} for $j=1 (2,3,4)$ and $i=1,2,...,8$, where $\lambda'_4=\alpha\lambda_0$. Up to now, we can divide users associated to a D2D TX (a relay, a BS, a cache-enabled user itself) into $8$ classes based on the $j$-th column of the matrix $\mathbf{D}$. The corresponding class request arrival rate and class traffic demand are respectively $\zeta_{i,j}=\frac{n_{i,j}\lambda_3\varsigma}{\lambda_0}$ and $\sigma_{i,j}=\frac{\zeta_{i,j} S}{\varrho}$ for $j=1 (2,3,4)$ and $i=1,2,...,8$. Similarly, let $x_{i,j}$ be the number of requests in class $i$ of the queue at a D2D TX (a relay, a BS, a cache-enabled user itself) for $j=1 (2,3,4)$. And let $x_j=(x_{1,j},x_{2,j},...,x_{8,j})$ be the vector counting the number of user requests in each class. The orthogonal transmission is assumed, where user requests arriving at a service node are served one after the other in a round-robin manner with equal portion of time. We may view a service node as a processor, where 8 classes of user requests with different arrival and service rates are queueing to be served. So the request arrivals and departures at a D2D TX (a relay, a BS, a cache-enabled user itself) can be regarded as a \emph{multiclass processor-sharing queue}.

Denote $\mathbb{D}:=\{1,2,...,8\}$ as the set of classes. Let the process $\{X_j(t); t\geq0\}$ describe the number of user requests in different classes of the queue at a D2D TX (a relay, a BS, a cache-enabled user itself) at time $t$ for $j=1 (2,3,4)$. Then $X_j(t)$ has discrete state space $\mathbb{N}^\mathbb{D}$ and is a continuous-time Markov process which can be generated by:
\begin{align}\label{generator}
\left\{
\begin{array}{ll}
q(x_j,x_j+\varepsilon_i)=\zeta_{i,j},&x_j\in\mathbb{N}^\mathbb{D}\\
q(x_j,x_j-\varepsilon_i)=\frac{\varrho  A_{i,j}}{S} \frac{x_{i,j}}{x_{\mathbb{D},j}},&x_j\in\mathbb{N}^\mathbb{D}, x_j>0.
\end{array}
 \right.
\end{align}
where $\varepsilon_i$ represents the vector of $\mathbb{N}^\mathbb{D}$ whose $i$-th element is 1 and 0 elsewhere. $x_{\mathbb{D},j}\triangleq\sum_{i\in\mathbb{D}}x_{i,j}$ represents the total number of users in the queue.  As to a steady network, the number of the requests leaving and arriving at the cell should be equal in the long run, i.e., the throughput is equal to the traffic demand. Define the throughput per  request as the ratio of the given throughput (i.e., the traffic demand) by the mean number of user requests for a steady system \cite{Queue}. Then,
\begin{proposition}\label{eee}
The mean number of user requests, the throughput per user request (Thr./Req.), and the delay of the $i$-th class ($i=1,2,...,8$) at a D2D TX (a relay, a BS, a cache-enabled user itself) for $j=1 (2,3,4)$ are respectively given by,
\begin{align}\label{lemmanew}
&\bar{N}_{i,j}=\frac{\sigma_{i,j}}{\left(1\!-\!\frac{\sigma_j}{\sigma_{c,j}}\right) A_{i,j}},~~~~\bar{T}_{i,j}=\left(1\!-\!\frac{\sigma_j}{\sigma_{c,j}}\right) A_{i,j},~~~~\bar{D}_{i,j}=\frac{1}{\left(1\!-\!\frac{\sigma_j}{\sigma_{c,j}}\right) A_{i,j}\varrho S^{-1}},
\end{align}
where $\sigma_j=\sum_{i=1}^{8}\sigma_{i,j}$ can be considered as the total traffic demand in the queue at a  service node. $\sigma_{c,j}=\frac{\sigma_j}{\sum_{i=1}^8\sigma_{i,j}A_{i,j}^{-1}}$ is a critical value such that the  queue will be at the steady state when $\sigma_j<\sigma_{c,j}$. And the  mean number of user requests, the Thr./Req., and the delay in the queue at a service node respectively are,
\begin{equation}\label{typical}
\bar{N}_j=\frac{\sigma_j}{\sigma_{c,j}-\sigma_j},~~~~~\bar{T}_j=\sigma_{c,j}-\sigma_j,~~~~~\bar{D}_j=\frac{\sigma_j}{(\sigma_{c,j}-\sigma_j)\bar{\zeta}_j},
\end{equation}
where $\bar{\zeta}_j=\sum_{i=1}^8\zeta_{i,j}$ is considered as the total traffic arrival rate of the queue at a service node.
\end{proposition}
\begin{proof}
The results can be deduced from \cite{Queue} and the references therein, and the proof is omitted in this paper to avoid the unnecessary repetition.
\end{proof}

Theorem \ref{rateltheorem} indicates that when the interference is dominant and $\alpha<\alpha^*$,  the service rates keep constant despite the change of the density of the BS/relay/D2D. It may lead to a misconception that the infrastructure can be deployed as scattered as possible. However, Proposition \ref{eee} reveals the critical condition to keep  the system steady, i.e., $\sigma_j<\sigma_{c,j}, \forall j=1,2,...,4$ should be satisfied when arranging the network. We call $\frac{\sigma_j}{\sigma_{c,j}}$ \emph{steady ruler} of the network. The critical condition decides the maximum load/throughput of the system, e.g., the maximum arrival rate ($\varsigma^*$) of the request,
\begin{equation}
\varsigma^*=\text{max}\left\{\varsigma~\begin{array}{|c}
                           \frac{\sigma_j}{\sigma_{c,j}}<1, \forall i\in\mathbb{D}, j=1,2,...,4
                         \end{array}\right\}.
\end{equation}
Besides, Proposition \ref{eee} points out the maximum ratio of $\frac{\lambda_0}{\lambda_i},i=1,2,3$ for the network planning. Apparently, smaller densities of network infrastructures means more user connection per BS/relay/D2D, which will destruct the steady state of the queue and lead to the request congestion.

Analyzing Proposition \ref{eee} we observe that, larger content size $S$, arrival rate $\varsigma$ and $\frac{1}{\varrho}$ are not helpful for the improvement of the performance.  Higher service rate  ($A_{i,j}$) is important for the smooth departure of requests, yielding the performance improvement in terms of the Thr./Req. and delay. Moreover, the network performance highly depends on the number of users in each class ($n_{i,j}$), which are determined by the transmission power ($P_i$), node density ($\lambda_i$), content popularity ($\gamma$), caching ability ($M_1, M_2$) and association protocol. Specifically, when the user is able to obtain contents from its local caching immediately, the Thr./Req. (the delay) tends to infinity (zero) for the fact that the value of $A_{7,4}=\mathcal{U}_l$ is extremely high.
\section{Numerical Results}
In this section, we simulate the cache-enabled network to verify the performance of the proposed system. We obtain the results with Monte Carlo methods in a square area of $2000\text{m} \times 2000\text{m}$, where the nodes are scattered based on independent homogeneous PPPs with intensities of   $\{\lambda_0,\lambda_2,\lambda_3\}=\{\frac{300}{\pi500^2},\frac{5}{\pi500^2},\frac{1}{\pi500^2}\}~\text{nodes/m}^2$. The transmit powers are $\{P_1,P_2,P_3\}=\{23,33,43\}$ dBm and 20 MHz bandwidth are shared among different tiers. We set the path-loss $\beta=4$, total number of contents $N=200$, the size of each content $S=100$ Mbits, the caching ability $M_1=5$ and $M_2=50$, and the content popularity $\gamma=0.8$.  These typical parameters do not change unless  additional statements are clarified.
\begin{figure}[t]
\centering
\includegraphics[width=3.5in]{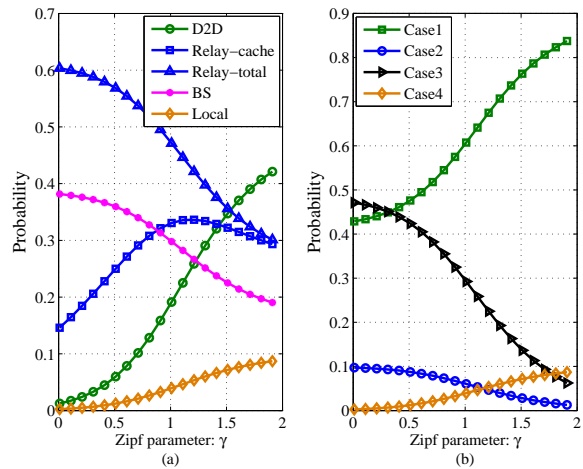}
\caption{The subfigure (a) illustrates the probability of users associated to different service nodes; the subfigure (b) illustrates the probability of users active in different Cases; $\alpha=0.1$.}
\label{probabilitycase}
\end{figure}

As illustrated in the subfigure (a) of Fig. \ref{probabilitycase}, the probability for the user to obtain contents from the D2D TX or its local caching space becomes higher with the increase of $\gamma$. Smaller fraction of users need to access the relay (Relay-total in the figure) or BS with more \textquotedblleft concentrated\textquotedblright~contents. It implies that cache-enabled network reduces the cell load of the BS tier and the relay tier. The number of users accessing content from the caching space of relays (Relay-cache in the figure) increases first and then deceases because of the traffic offloading ability of the D2D tier. In the subfigure (b), the probabilities of Case 1 and Case 4 increase with $\gamma$ as more contents can be obtained via the D2D link or from the local caching, yielding the increase of the cache hit rate.

The theoretical estimates and simulating results of the the average ergodic rates in Case 1-3 are illustrated in Fig. \ref{rate}, and they are consistent well. We obtain $\alpha^*=13.78\%$ and $\widehat{\alpha}=21.96\%$ with the parameters in  subfigure (a) based on (\ref{alpha}).  The rate in  Case 1 keeps constant when $\alpha$ changes in the region $[0, \alpha^*]$ and it starts to increase obviously from $\alpha^*$. It highly depends on whether all of the cache-enabled users are active as D2D TXs. We observe that the number of active D2D TXs increases linearly with $\alpha$ when $\alpha<\alpha^*$ as all of the cache-enabled users need to be active as D2D TXs. Only a part of cache-enabled users are active as D2D TXs when $\alpha>\alpha^*$ and it comes to the maximum number when $\alpha=\widehat{\alpha}$. As to Case 2, the rate decreases with the increase of $\alpha$ in the region $\alpha<\widehat{\alpha}$ owing to the increase of D2D TXs, then it increases slightly after $\widehat{\alpha}$ for the number of active D2D TXs goes down. The rates in Case 2 and Case 3 are smaller than those in Case 1. It is because the users of Case 2 and Case 3 can not obtain any benefit except for the unnecessary interference from D2D TXs. With the parameters in subfigure (b), we get $\alpha^*=0$ and $\widehat{\alpha}=31.62\%$. It means the transmit power and coverage of the D2D link are limited, thus only parts of cache-enabled users are active as D2D TXs for an arbitrary $\alpha$. The number of active D2D TXs increases with the increase of $\alpha$ when $\alpha\!<\!\widehat{\alpha}$. Consequently, the rate of Case 1 increases but those of Case 2 decrease with the increase of $\alpha$ in subfigure (b).
\begin{figure}[t]
\centering
\includegraphics[width=3.5in]{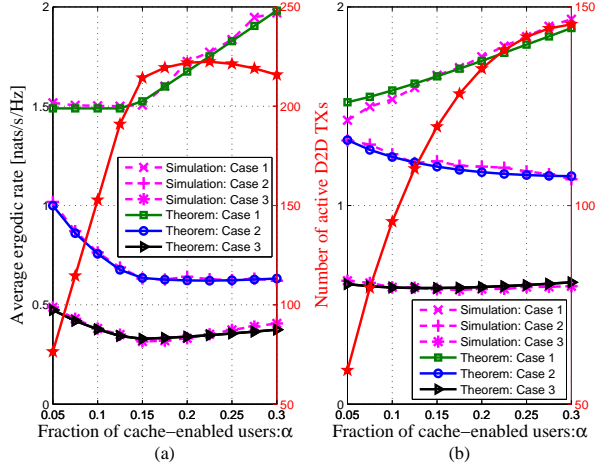}
\caption{The average ergodic rates of different Cases: the transmit powers of the subfigure (a) are $\{P_1,P_2,P_3\}=\{23,33,43\}~\text{dBm}$, those of the subfigure (b) are $\{P_1,P_2,P_3\}=\{13,33,43\} ~\text{dBm}$. The left and right ordinate (black and red) respectively correspond to the average ergodic rates and the number of active D2D TXs.}
\label{rate}
\end{figure}
\begin{figure}[t]
\centering
\includegraphics[width=3.5in]{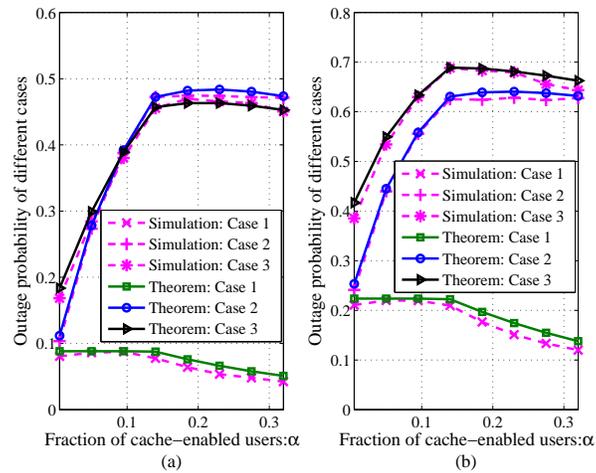}
\caption{The outage probabilities of different Cases: the SINR threshold of the subfigure (a) and (b) are $\tau=-10\text{dB}$ and $\tau=-5\text{dB}$, respectively.}
\label{outage}
\end{figure}

The outage probabilities of different Cases are demonstrated in Fig. \ref{outage}. The outage probability decreases with the decrease of SINR target $\tau$, where lower SINR target means more interference is allowed. Similar to the average rate, the outage probability of Case 1 keeps constant before $\alpha$ goes to $\alpha^*$ and then decreases obviously.  Case 2 and Case 3 have higher outage probability compared with Case 1. Fig. \ref{outage} can be explained from another perspective with Fig. \ref{CDF} where the CDF of the SINR are demonstrated.  As an example, for SINR $=-10$ dB in Fig. \ref{CDF}, the value of the CDF of Case 2 is smaller than that of Case 3 when $\alpha=0.05$, while the former approximately equals to the latter when $\alpha=0.1$. It conforms to what is illustrated in Fig. \ref{outage}. Moreover, both $0.05$ and $0.1$ are smaller than $\alpha^*=13.78\%$, so the CDF of SINR for Case 1 when $\alpha=0.05$ coincides with that when $\alpha=0.1$ in Fig. \ref{CDF}.
\begin{figure}[t]
\centering
\includegraphics[width=3.5in]{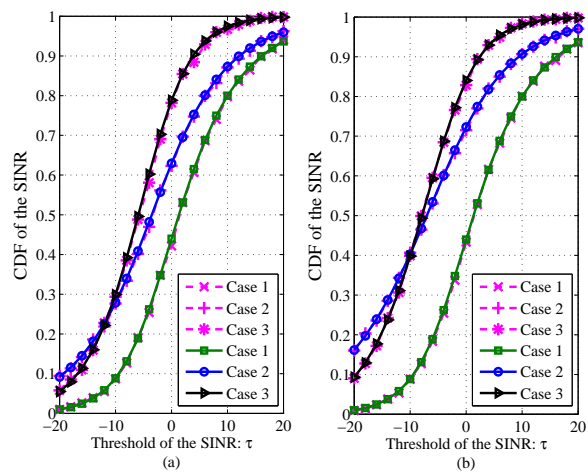}
\caption{The CDF of SINR for different Cases: the fraction of cache-enabled user in the subfigure (a) and (b) are $\alpha=0.05$ and $\alpha=0.10$, respectively.}
\label{CDF}
\end{figure}

Fig. \ref{throughtputper} compares the Thr./Req. of the cache-enabled network with that of the baseline. As in the subfigure (a), the Thr./Req. of class 1 (class 3) at the BS is higher than (approximately equals to) that in the baseline. Meanwhile, because of the strong unnecessary interference triggered by the D2D TXs, class 5 has lower Thr./Req. than that of the baseline. The other virtual classes whose Thr./Req. are zero are not illustrated in the figure. Similarly, The performance of the $j$-th $(j=1,...,4)$ class at the relay is better than or approximately equals to that in the baseline, while class 5 and class 6 get worse. Additional process is needed for the relay to fetch the uncached contents via the wired backhaul link, so the Thr./Req. of class 4 (class 6) is smaller than that of class 3 (class 5). In Fig.  \ref{throughtputper} (c), we compare the Thr./Req. at the D2D TX with that at the BS in the baseline. The Thr./Req. of D2D TX outperforms that of BS in the baseline by 46.8\%-58.1\%.  D2D TXs give rise to interference, yet at the same time traffic loads of BSs and relays are offloaded by D2D TXs and the caching resources. Consequently, the Thr./Req. in the queue at the relay and BS are not seriously affected by the interference, while the throughput over the entire network increases significantly because of the increase of the number of the simultaneously active nodes and the tolerable request arrival rate.

We present the \emph{steady ruler} versus the request arrival rate in Fig. \ref{steady2} to evaluate the throughput gain of the network. Based on the value of the \emph{steady ruler} $\frac{\sigma_j}{\sigma_{c,j}}$ for different typical queues, we circle out the critical point for the network to keep steady. From the figure we can see that the \emph{steady ruler} of the relay and the D2D are far smaller than that of the BS. So the maximum load of the network, e.g., the maximum request arrival rate $\varsigma^*$, is decided by the state of the queue at the BS. The BS has a wider range of the coverage compared with that of the relay and the D2D owning to the higher transmit power. Most of users are covered by the BS and join in the queue at the BS. According to the critical point, we observe that when $\gamma=$ 0.8~(1.8) the throughput gain over the entire network is 13.3\% (57.3\%) compared with that of the baseline. Moreover, the \emph{steady ruler} of the relay and the D2D is smaller than 1, so the relay can be deployed in the high-density area, and more opportunity should be given to the user to access the content via the D2D link, yielding the efficient offloading of the cellular traffic.
\begin{figure*}[t]
\centering
\includegraphics[width=6.0in]{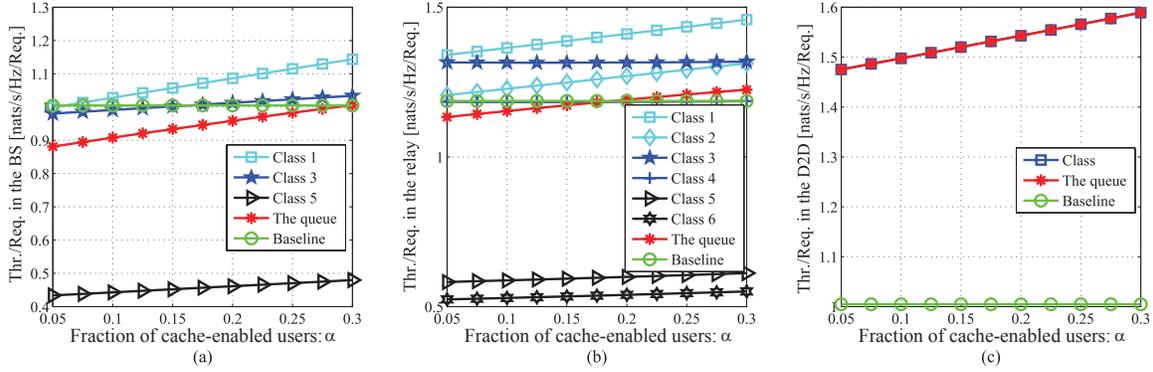}
\caption{The subfigure (a), (b) and (c) respectively illustrate the Thr./Req. at the BS, relay and D2D TXs: $\{P_1,P_2,P_3\}=\{13,33,43\}~\text{dBm},\{\lambda_0,\lambda_2,\lambda_3\}=\{\frac{300}{\pi500^2},\frac{30}{\pi500^2},\frac{6}{\pi500^2}\}~\text{nodes/m}^2$, $\varsigma=0.25, \varrho=1$.}
\label{throughtputper}
\end{figure*}
\begin{figure}[t]
\centering
\includegraphics[width=3.5in]{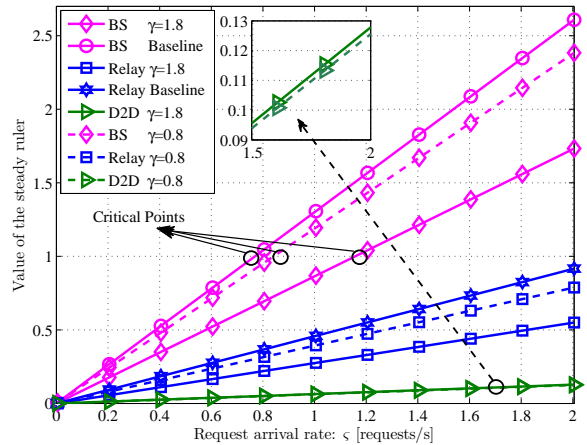}
\caption{The throughput gain for the cache-enabled network compared with that of the baseline: $\{P_1,P_2,P_3\}=\{13,33,43\}~\text{dBm},\{\lambda_0,\lambda_2,\lambda_3\}=\{\frac{300}{\pi500^2},\frac{30}{\pi500^2},\frac{6}{\pi500^2}\}~\text{nodes/m}^2$, $\varsigma=0.25, \varrho=1,\alpha=0.25$.}
\label{steady2}
\end{figure}
\begin{figure}[t]
\centering
\includegraphics[width=3.5in]{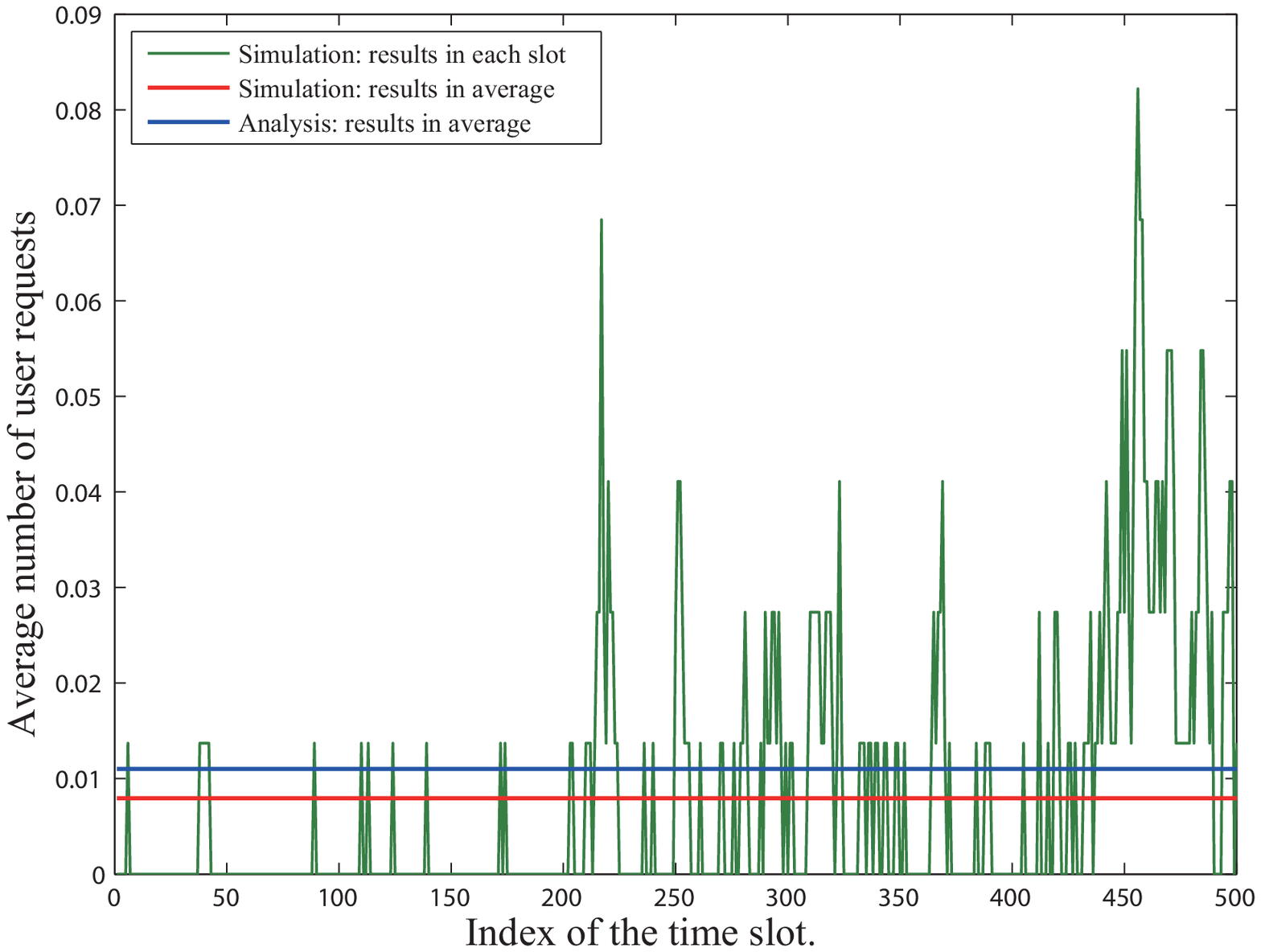}
\caption{The number of user requests in the queue of the D2D TX: $\{P_1,P_2,P_3\}=\{13,33,43\}~\text{dBm},\{\lambda_0,\lambda_2,\lambda_3\}=\{\frac{300}{\pi500^2},\frac{30}{\pi500^2},\frac{6}{\pi500^2}\}~\text{nodes/m}^2$, $\varsigma=0.25,\varrho=1,\alpha=0.25$, slot-step $=0.2$ second.}
\label{N_D2D}
\end{figure}

For further discussions, we divide the time into slots with equal duration for content transmission. Requests with different volumes of contents are responded with corresponding rate in several slots. In the simulation, we choose 500 slots of them to investigate the number of user requests at the D2D TX in each time slot, based on which the average number of user requests during the 500 slots are also illustrated in Fig. \ref{N_D2D}. We observe that the average number of user requests in the simulation is lower than that of the analysis as dummy transmissions are assumed in the analysis. More precise analysis  can be a promising topic for the further work and the analysis result in this paper is a lower bound of the performance for the cache-enabled network.
\section{conclusion}
The paper aims to model and evaluate the performance of the wireless HetNet where the RAN caching and D2D caching coexist. The caching ability is available in both the relay and some of the users. We propose to cache the most popular multimedia contents via broadcasting during off-peak time to be reused for frequent access. Firstly, we model the node locations of the HetNet as  mutually independent PPPs. According to the maximum received-power cell association scheme, users can flexibly connect to the cellular and D2D link. Users are classified into four Cases according to whether the requesting user is cache-enabled and the type of the service node. We theoretically elaborate the average ergodic rates and the outage probabilities of  different Cases in the downlink. Based on the Case the user is active in, the user requests arriving at a D2D TX (relay, BS, cache-enabled user itself) can be classified into different classes. The throughput and the delay of different classes are then derived with modeling the \emph{multiclass processor-sharing queue} and the continuous-time Markov process. We further provide the \emph{steady ruler} for the HetNet, which decides the maximum traffic load/throughput of the network. Numerical results show that the global throughput of the cache-enabled system can increase by 57.3\% compared with that of the system without caching ability.
\begin{appendix}
\subsection{Proof of Proposition \ref{ccc}}
With extension of (\ref{r1r2r3}), the joint PDF of the distance $R_1,R_2,...,R_K$ is
\begin{equation}
f_{R_1,...,R_K}(r_1,...,r_K)=\left(\prod_{i=1}^{K}2\pi\lambda_{i}r_{i}\right)e^{-\pi\sum\limits_{i=1}^{K}\lambda_i r_{i}^{2}}.
\end{equation}
The Euclidean region $\Omega$ satisfying $C_{t_1}\!>\!C_{t_2}\!>\!...\!>\!C_{t_{K-1}}\!>\!C_{t_K}$ is
\begin{equation}
\Omega = \left\{
\begin{array}{rl}
0<r_{t_1}<+\infty,\\
\left(\frac{P_{t_{2}}}{P_{t_1}}\right)^\frac{1}{\beta}r_{t_1}<r_{t_{2}}<+\infty,\\
...~~~~~~~~~~~\\
\left(\frac{P_{t_{K}}}{P_{t_{K-1}}}\right)^\frac{1}{\beta}r_{t_{K-1}}<r_{t_{K}}<+\infty.
\end{array}
 \right.
\end{equation}
So the probability of $C_{t_1}\!>\!C_{t_2}\!>\!...\!>\!C_{t_{K-1}}\!>\!C_{t_K}$ is
\begin{align}\
&\mathbb{P}(C_{t_1}\!>\!C_{t_2}\!>\!...\!>\!C_{t_{K-1}}\!>\!C_{t_K})=\iint\!\!...\!\!\int_{\Omega}f_{R_{t_1}\!,R_{t_{2}},...,R_{t_K}}\!(r_{t_1},r_{t_{2}},...,r_{t_K})\mathrm{d}{r_{t_K}}\!\mathrm{d}{r_{K-1}}...\mathrm{d}{r_1}
\nonumber\\
&=\iint\!\!...\!\!\int_{\Omega}\left(\prod_{i=1}^{K}2\pi\lambda_{t_i}r_{t_i}\right)e^{-\pi\sum\limits_{i=1}^{K}\lambda_{t_i}r_{t_i}^{2}}\mathrm{d}{r_{t_K}}\!\mathrm{d}{r_{K-1}}...\mathrm{d}{r_1}\stackrel{(a)}{=}\prod_{n=1}^{K-1}\left[\sum_{m=n}^K\frac{\lambda_{t_m}}{\lambda_{t_n}}\left(\frac{P_{t_m}}{P_{t_n}}\right)^\frac{2}{\beta}\right]^{-1},
\end{align}
where (a) follows when integrating with the region of $\Omega$, and the proof is completed. \hfill\ensuremath{\blacksquare}
\subsection{Proof of Lemma \ref{ddd}}
In this Case, we only need to ensure that $C_i$ is higher than that of any other tiers, and the order of the other tiers need not to be cared. So we have
\begin{align}\label{pp}
&\mathcal{G}_{K,i}\triangleq \mathbb{P}(C_{i}\!>\!\max_{\forall n\neq i}C_{n})=\mathbb{P}(C_{i}\!>C_{1},...,C_{i}\!>C_{i-1},C_{i}\!>C_{i+1},...,C_{i}\!>C_{K})\nonumber\\
&=\iint\!\!...\!\!\int_{\Omega}f_{R_{1},...,R_K}\!(r_{1},...,r_K)\mathrm{d}{r_{K}}\mathrm{d}{r_{K-1}}...\mathrm{d}{r_1}
=\left[\sum_{m=j}^K\frac{\lambda_{t_m}}{\lambda_{t_j}}\left(\frac{P_{t_m}}{P_{t_j}}\right)^\frac{2}{\beta}\right]^{-1},
\end{align}
where the integral region $\Omega$ now is
\begin{equation}
\Omega = \left\{
\begin{array}{rl}
0<r_{i}<+\infty,\\
\left(\frac{P_{1}}{P_i}\right)^\frac{1}{\beta}r_{i}<r_{1}<+\infty,\\
...~~~~~~~~~~~\\
\left(\frac{P_{i-1}}{P_i}\right)^\frac{1}{\beta}r_{i}<r_{i-1}<+\infty,\\
\left(\frac{P_{i+1}}{P_i}\right)^\frac{i}{\beta}r_{i}<r_{i+1}<+\infty,\\
...~~~~~~~~~~~\\
\left(\frac{P_{K}}{P_j}\right)^\frac{1}{\beta}r_{i}<r_{{K}}<+\infty.
\end{array}
 \right.
\end{equation}
Then we get the lemma. \hfill\ensuremath{\blacksquare}
\subsection{Proof of Theorem 1}\label{sec:prooftheroem1}
According to (\ref{rate2}), we have
\begin{align}\label{ui}
\mathcal{U}_{1,i}&=
\int_0^\infty\mathbb{E}_{\text{SINR}_i}\left[\begin{array}{c|}
                          \!\! \text{ln}(1+\text{SINR}_i(x))
                         \end{array}~x\right]f_{X_i}(x)\mathrm{d}{x}.
\end{align}
Because of $\mathbb{E}[X]=\int_0^\infty \mathbb{P}(X>t)\mathrm{d}{t}$ when $X>0$, we have
\begin{align}\label{esinr}
\mathbb{E}_{\text{SINR}_i}\left[\begin{array}{c|}
                           \!\!\text{ln}(1+\text{SINR}_i(x))
                         \end{array}~x\right]&=\int_0^\infty \mathbb{P}\left[\begin{array}{c|}
                           \!\!\text{ln}(1+\text{SINR}_i(x))>t
                         \end{array}~x\right]\mathrm{d}{t}
\nonumber\\
&=\int_0^\infty \mathbb{P}\left[\begin{array}{c|}
                           \!\!g_{i,0}>x^\beta P_i^{-1}I_r(e^t-1)
                         \end{array}~x\right]\mathrm{d}{t},
\end{align}
where
\begin{align}\label{pgi0}
&\mathbb{P}\left[\begin{array}{c|}
                           \!\!g_{i,0}{\setlength\arraycolsep{0.5pt}>}x^\beta P_i^{-1}I_r(e^t{\setlength\arraycolsep{0.5pt}-}1)
        \end{array}~x\right]{\setlength\arraycolsep{0.5pt}=}\mathbb{E}_{I_r}\left[\mathbb{P}\left[\begin{array}{c|}
                           \!\!g_{i,0}{\setlength\arraycolsep{0.5pt}>}x^\beta P_i^{-1}I_r(e^t{\setlength\arraycolsep{0.5pt}-}1)
                         \end{array}~I_r,x\right]\right]{{\setlength\arraycolsep{0.5pt}=}}\mathbb{E}_{I_r}\left[\begin{array}{c|}
                           e^{-x^\beta P_i^{-1}I_r(e^t-1)}
                         \end{array}~x\right]
\nonumber\\
&{=}e^{-x^\beta P_i^{-1}(e^t-1)\sigma^2}\mathbb{E}\left[\begin{array}{c|}
                           {e^{-x^\beta P_i^{-1}(e^t-1)\sum\limits_{j=1}^3I_j}}
                         \end{array}~x\right]\stackrel{(a)}{=}e^{-x^\beta P_i^{-1}(e^t-1)\sigma^2}\prod_{j=1}^3\mathcal{L}_{I_j}\left[x^\beta P_i^{-1}(e^t-1)\right],
\end{align}
where Step (a) follows the fact that $\Phi_i$ are mutually independent PPPs. Here, the interference comes from the actually active nodes in the $i$-th tier with density $\lambda'_i$ for $i=1,2,3$. So the Laplace transform $\mathcal{L}_{I_j}\left[x^\beta P_i^{-1}(e^t-1)\right]$ is
\begin{align}\label{laplace1}
&\mathcal{L}_{I_j}\left[x^\beta P_i^{-1}(e^t-1)\right]=\mathbb{E}_{I_j}\left[e^{-x^\beta P_i^{-1}(e^t-1)I_j}\right]
=\mathbb{E}_{\Phi_j}\left[e^{-x^\beta P_i^{-1}(e^t-1)\sum_{k\in\phi_j\setminus B_{i,0}}P_jh_{jk}|Y_{jk}|^{-\beta}}\right]\nonumber\\
&{=}e^{-2\pi\lambda'_j\int_{z_j}^\infty\left[1-\mathcal{L}_{h_{jk}}(x^\beta\frac{P_j}{P_i}(e^t-1)y^{-\beta})\right]y\mathrm{d}{y}}
=e^{-2\pi\lambda'_j\int_{z_j}^\infty\left[1-({1+x^\beta\frac{P_j}{P_i}(e^t-1)y^{-\beta}})^{-1}\right]y\mathrm{d}{y}}\\
&{=}e^{-2\pi\lambda'_j\mathlarger{\int}_{z_j}^\infty\frac{y}{1+[x^\beta\frac{P_j}{P_i}(e^t-1)]^{-1}y^{\beta}}\mathrm{d}{y}}
\stackrel{(a)}{=}e^{[-\pi\lambda'_j\left(\frac{P_j}{P_i}\right)^{\frac{2}{\beta}}\!x^2\!(e^t\!-\!1)^\frac{2}{\beta}\mathlarger{\int}_{(e^t\!-\!1)^{-\frac{2}{\beta}}}^\infty({1+u^{\frac{\beta}{2}}})^{-1}\mathrm{d}{u}]}
=e^{-\pi\lambda'_j\left(\frac{P_j}{P_i}\right)^{\frac{2}{\beta}}x^2\mathcal{Z}_1(e^t-1)},\nonumber
\end{align}
where $z_j=x(\frac{P_j}{P_i})^{\frac{1}{\beta}}$ is the distance between the reference user and its closest interference node. Meanwhile, by using change of variables with $u=[x^\beta\frac{P_j}{P_i}(e^t-1)]^{-\frac{2}{\beta}}y^2$, we have Step (a). In the expression above, we use
$\mathcal{Z}_1(e^t\!-\!1)=\frac{2(e^t-1)}{\beta-2}{}_2F_1[1,1-\frac{2}{\beta};2-\frac{2}{\beta};1-e^t]$,
where ${}_2F_1[\cdot]$ denotes the Gauss hypergeometric function. Accordingly, we have
\begin{align}\label{ui_z1}
\!\!\!\mathcal{U}_{1,i}&{\setlength\arraycolsep{0.5pt}=}
\frac{2\pi \lambda_i}{\mathcal{G}_{3,i}}\int_0^\infty\!\!\!\int_0^\infty x \text{exp}\left\{-x^\beta P_i^{-1}(e^t-1)\sigma^2{\setlength\arraycolsep{0.5pt}-}
\pi x^2\sum\limits_{j=1}^3\left[\lambda_j{\setlength\arraycolsep{0.5pt}+}
\lambda'_j\mathcal{Z}_1(e^t{\setlength\arraycolsep{0.5pt}-}
1)\right]\left(\frac{P_j}{P_i}\right)^{\frac{2}{\beta}}\right\}\mathrm{d}{t}\mathrm{d}{x}.
\end{align}
Using (\ref{prefer}), $\sum_{j=1}^3\lambda_j\left(\frac{P_j}{P_i}\right)^{\frac{2}{\beta}}=\frac{\lambda_i}{\mathcal{G}_{3,i}}$, and $\lambda_2=\lambda'_2,\lambda_3=\lambda'_3$, we get (\ref{ui_z}). \hfill\ensuremath{\blacksquare}
\subsection{Proof of Theorem 2}\label{sec:prooftheroem2}
Referring to the analysis of Theorem \ref{rateltheorem}, we have
\begin{align}\label{ui_2}
\mathcal{U}_{2,i}=\!\int_0^\infty\!\!\!\int_0^\infty\!\!\!e^{-\frac{\sigma^2x^\beta(e^t-1)}{P_i}}\!\!\prod_{j=1}^3\!\!\mathcal{L}_{I_j}\!\!\left[x^\beta P_i^{-1}(e^t\!-\!1)\right]\!\widehat{f}_{X_i}(x)\mathrm{d}{t}\mathrm{d}{x}.
\end{align}
For $j=1$, the distance between the reference user and its closest interfering cache-enabled user can be as close as 0. The Laplace transform $\mathcal{L}_{I_1}\left[x^\beta P_i^{-1}(e^t-1)\right]$ can be computed as
\begin{align}\label{laplace1_2}
\mathcal{L}_{I_1}\left[x^\beta P_i^{-1}(e^t-1)\right] &=e^{-2\pi\lambda'_1\mathlarger{\int}_{\widehat{a}}^\infty\left[1-\mathcal{L}_{h_{1k}}(x^\beta\frac{P_1}{P_i}(e^t-1)y^{-\beta})\right]y\mathrm{d}{y}}=e^{-\pi\lambda'_1\left(\frac{P_1}{P_i}\right)^{\frac{2}{\beta}}x^2\mathcal{Z}_2(a)},
\end{align}
where $\widehat{a}~ \text{and}~ a$ are variables as small as $0$ and
$\mathcal{Z}_2(a)\!=\!(e^t\!-\!1)^{\frac{2}{\beta}}\frac{2a^{\frac{2-\beta}{2}}}{\beta-2}{}_2F_1[1,1-\frac{2}{\beta};2-\frac{2}{\beta};-a^{\frac{-\beta}{2}}]$. Plugging (\ref{laplace1}), (\ref{fx2}) and (\ref{laplace1_2}) into (\ref{ui_2}), we obtain the average ergodic rate of the reference user in (\ref{ui_z_2}) based on $\sum_{j=2}^3\lambda_j(\frac{P_j}{P_i})^{\frac{2}{\beta}}=\frac{\lambda_i}{\mathbb{P}_{i,j}}$ and
$[\sum_{j=2}^3\frac{\lambda_j}{\lambda_1}(\frac{P_j}{P_1})^{\frac{2}{\beta}}]^{-1}=\frac{\mathcal{G}_{3,1}}{1\!-\!\mathcal{G}_{3,1}}$.\hfill\ensuremath{\blacksquare}
\subsection{Proof of Joint PDF of $X_1,Y_j$}\label{sec:fxy}
The joint probability of $ 0<X_1<x,0<Y_j<y,(y>(\frac{P_j}{P_1})^{\frac{1}{\beta}}x)$ is
\begin{align}
   &\mathbb{P}(0<X_1<x,0<Y_j<y)=\mathbb{P}(\begin{array}{c|c}
                                    0<r_1<x,0<r_j<y&C_{1,j,k}
                                  \end{array})\nonumber\\
   &{=}\mathbb{P}\left(0<\!r_1\!<x, (\frac{P_j}{P_1})^{\frac{1}{\beta}}r_1<\!r_j\!<y,(\frac{P_k}{P_j})^{\frac{1}{\beta}}r_j<\!r_k\!<\infty\right)\mathbb{P}_{1,j,k}^{-1}\nonumber\\
   &{=}\mathbb{P}_{1,j,k}^{-1}\int_0^x\int_{(\frac{P_j}{P_1})^{\frac{1}{\beta}}r_1}^y\int_{(\frac{P_k}{P_j})^{\frac{1}{\beta}}r_j}^\infty f_{R_1,R_j,R_k}(r_1,r_j,r_k)\mathrm{d}{r_k}\mathrm{d}{r_j}\mathrm{d}{r_1}.
\end{align}
Then the joint PDF of $x,y$ is
\begin{align}\label{fxy}
&f_{X_1,Y_j}(x,y)=\frac{\partial \mathbb{P}(X_1<x,Y_j<y)}{\partial x \partial y}=\left\{
\begin{array}{cc}
\frac{4\pi^2\lambda_1\lambda_jxy\mathbb{P}_{1,j,k}^{-1}}{e^{\pi\lambda_1x^2+\pi\lambda_jy^2\left[1+\frac{\lambda_k}{\lambda_j}\left(\frac{P_k}{P_j}\right)^{\frac{2}{\beta}}\right]}},&y>(\frac{P_j}{P_1})^{\frac{1}{\beta}}x\\
~~~~~~~~~~~~0~~~~~~~~~~~~~~~,&\text{else.}\\
\end{array}
 \right.
\end{align}
\hfill\ensuremath{\blacksquare}
\subsection{Proof of Theorem 3}\label{sec:prooftheroem3}
In this Case, (\ref{ui}) can be rewritten as
\begin{align}
\mathcal{U}_{3,j}&=\!\!\!
\int_0^\infty\!\!\!\!\int_0^\infty\!\!\!\mathbb{E}_{\text{SINR}_j}\!\left[\begin{array}{c|}
                          \!\! \text{ln}(1+\text{SINR}_j(y))
                         \end{array}~x,y\right]\!f_{X_1,Y_j}\!(x,y)\mathrm{d}{x}\mathrm{d}{y}.
\end{align}
Similar to (\ref{esinr}), we have
\begin{align}
&\mathbb{E}_{\text{SINR}_j}\left[\begin{array}{c|}
                           \!\!\text{ln}(1+\text{SINR}_j(y))
                         \end{array}~x,y\right]=\int_0^\infty \mathbb{P}\left[\begin{array}{c|}
                           \!\!g_{j,0}>y^\beta P_j^{-1}I_r(e^t-1)
                         \end{array}~x,y\right]\mathrm{d}{t},
\end{align}
where
\begin{align}
&\mathbb{P}\left[\begin{array}{c|}
                           \!\!g_{j,0}>y^\beta P_j^{-1}I_r(e^t-1)
        \end{array}~x,y\right]
=e^{-y^\beta P_j^{-1}(e^t-1)\sigma^2}\prod_{i=1}^3\mathcal{L}_{I_i}\left[y^\beta P_j^{-1}(e^t-1)\right].
\end{align}
Similarly, the Laplace transform of $I_i$ for $i=2, 3$ is
\begin{align}
&\mathcal{L}_{I_i}\left[y^\beta P_j^{-1}(e^t-1)\right]=e^{-\pi\lambda_i\left(\frac{P_i}{P_j}\right)^{\frac{2}{\beta}}y^2\mathcal{Z}_1(e^t-1)}.
\end{align}
And for $i=1$, we have
\begin{align}
&\mathcal{L}_{I_1}\left[y^\beta P_j^{-1}(e^t\!-\!1)\right]=e^{\!-\!2\pi\lambda'_1\mathlarger{\int}_x^\infty\frac{v\mathrm{d}{\upsilon}}{1+[y^\beta\frac{P_1}{P_j}(e^t\!-\!1)]^{\!-\!1}v^\beta}}=e^{-\pi\lambda'_1y^2[\frac{P_1}{P_j}(e^t-1)]^{\frac{2}{\beta}}\mathlarger{\int}_{x^2[y^\beta\frac{P_1}{P_j}(e^t-1)]^{-\frac{2}{\beta}}}^\infty\frac{\mathrm{d}{u}}{1+u^{\frac{\beta}{2}}}}.
\end{align}
Considering $\sigma^2\rightarrow0$, we have
\begin{align}\label{zuichang}
&\mathcal{U}_{3,j}{\setlength\arraycolsep{0.5pt}=}
\!\!\!\int_0^\infty\!\!\!
\int_0^\infty\!\!\!\int_0^{(\frac{P_1}{P_j})^{\frac{1}{\beta}}y}\frac{\prod\limits_{i=j,k}\mathcal{L}_{I_i}\left[y^\beta P_j^{-1}(e^t{\setlength\arraycolsep{0.5pt}-}
1)\right]}{\text{exp}\{{
\pi\lambda'_1y^2[\frac{P_1}{P_j}(e^t-1)]^{\frac{2}{\beta}}\int_{[y^\beta\frac{P_1}{P_j}(e^t-1)]^{-\frac{2}{\beta}}x^2}^\infty\frac{1}{1+u^{\frac{\beta}{2}}}\mathrm{d}{u}}\}}f_{X_1,Y_j}(x,y)\mathrm{d}{x}\mathrm{d}{y}\mathrm{d}{t}\nonumber\\
&\stackrel{(a)}{=}\!\!\!\int_0^\infty\!\!\!
\int_0^\infty\!\!\!\int_0^{1}\frac{\prod\limits_{i=j,k}\mathcal{L}_{I_i}\left[y^\beta P_j^{-1}(e^t{\setlength\arraycolsep{0.5pt}-}
1)\right]}{\text{exp}\{{
\pi\lambda'_1y^2[\frac{P_1}{P_j}(e^t{\setlength\arraycolsep{0.5pt}-}
1)]^{\frac{2}{\beta}}\int_{(e^t-1)^{{\setlength\arraycolsep{0.5pt}-}
\frac{2}{\beta}}x^2}^\infty\frac{1}{1+u^{\frac{\beta}{2}}}\mathrm{d}{u}}\}}f_{X_1,Y_j}\left((\frac{P_1}{P_j})^\frac{1}{\beta}yx,y\right)(\frac{P_1}{P_j})^\frac{1}{\beta}y\mathrm{d}{x}\mathrm{d}{y}\mathrm{d}{t}\nonumber\\
&=\!\!\!\int_0^\infty\!\!\!\int_0^{1}\frac{2\lambda_1\lambda_j\mathbb{P}_{1,j,k}^{-1}(\frac{P_1}{P_j})^{\frac{2}{\beta}}x}{\left\{\lambda_j\left[1+\frac{\lambda_k}{\lambda_j}(\frac{P_k}{P_j})^{\frac{2}{\beta}}\right]\left[1+\mathcal{Z}
_1(e^t-1)\right]+\left[\lambda'_1\mathcal{Z}_3(e^t-1)+\lambda_1\right](\frac{P_1}{P_j})^{\frac{2}{\beta}}x^2\right\}^2}\mathrm{d}{x}\mathrm{d}{t},
\end{align}
where $\mathcal{Z}_3(e^t-1)=\frac{2(e^t-1)}{\beta-2}x^{-\beta}{}_2F_1[1,1-\frac{2}{\beta};2-\frac{2}{\beta};(1-e^t)x^{-\beta}]$. By using a change of variables $x=(\frac{P_1}{P_j})^{\frac{1}{\beta}}yx'$, we obtain Step (a). Because $[\sum\limits_{j=2}^3\frac{\lambda_j}{\lambda_1}(\frac{P_j}{P_1})^{\frac{2}{\beta}}]^{-1}=\frac{\mathcal{G}_{3,1}}{1\!-\!\mathcal{G}_{3,1}}$, (\ref{zuichang}) can be simplified to
\begin{align}\label{u3}
\mathcal{U}_{3,j}=\int_0^\infty\!\!\!\!\int_0^{1}\!\!\!\frac{\frac{2x\mathcal{G}_{3,1}}{1\!-\!\mathcal{G}_{3,1}}\mathbb{P}_{1,j,k}^{-1}\mathbb{P}_{j,k}\mathrm{d}{x}\mathrm{d}{t}}{\!\left\{1\!+\!\mathcal{Z}
_1(e^t\!-\!1)\!+\!\frac{\mathcal{G}_{3,1}x^2}{1\!-\!\mathcal{G}_{3,1}}\left[1\!+\!\frac{\lambda'_1}{\lambda_1}\mathcal{Z}_3(e^t\!-\!1)\right]\!\right\}^2}.
\end{align}
Based on $\mathbb{P}_{1,j,k}^{-1}\mathbb{P}_{j,k}=\mathcal{G}_{3,1}^{-1}$, (\ref{u3}) turns to (\ref{rate3}) and the proof is finished. \hfill\ensuremath{\blacksquare}
\subsection{Proof of Theorem 4}\label{sec:prooftheroem4}
Accordingly, it is easy to obtain
\begin{align}
\mathcal{P}_{1,i}&=1-\mathbb{E}[\mathbb{P}[\text{SINR}_i(x)>\tau]]=1-\int_0^\infty \mathbb{P}\left[\text{SINR}_i(x)\begin{array}{c|}
                          \!\!>\tau
                         \end{array}~x\right]f_{X_i}(x)\mathrm{d}{x}.
\end{align}
Similar to (\ref{pgi0}) and (\ref{laplace1}), we have
\begin{align}
\!\!\!\mathbb{P}\left[\text{SINR}_i(x)\begin{array}{c|}
                          \!\!>\tau
                         \end{array}~x\right]&=\mathbb{P}\left[\begin{array}{c|}
                           \!\!g_{i,0}>x^\beta P_i^{-1}I_r\tau
        \end{array}~x\right]=e^{-x^\beta P_i^{-1}\tau\sigma^2}\prod_{j=1}^3\mathcal{L}_{I_j}\left[x^\beta P_i^{-1}\tau\right],
\end{align}
where $\mathcal{L}_{I_j}[x^\beta P_i^{{\setlength\arraycolsep{0.5pt}-}1}\tau]{\setlength\arraycolsep{0.5pt}=}e^{{\setlength\arraycolsep{0.5pt}-}\pi\lambda'_j(\frac{P_j}{P_i})^{\frac{2}{\beta}}x^2\mathcal{Z}_1(\tau)}$,  $\mathcal{Z}_1(\tau){\setlength\arraycolsep{0.5pt}=}\frac{2\tau}{\beta{\setlength\arraycolsep{0.5pt}-}2}{}_2F_1[1,1{\setlength\arraycolsep{0.5pt}-}\frac{2}{\beta};2{\setlength\arraycolsep{0.5pt}-}\frac{2}{\beta};{\setlength\arraycolsep{0.5pt}-}\tau]$.
Then we get (\ref{oi_z}). \hfill\ensuremath{\blacksquare}
\subsection{Proof of Theorem 5}\label{sec:prooftheroem5}
According to (\ref{outage1}), the average outage probability is given by
\begin{align}
\!\!\!\mathcal{P}_{2,i}&{\setlength\arraycolsep{0.5pt}=}1{\setlength\arraycolsep{0.5pt}-}\!\!\int_0^\infty \mathbb{P}[\text{SINR}_i(x)\begin{array}{c|}
                          \!\!{\setlength\arraycolsep{0.5pt}>}\tau
                         \end{array}~x]\widehat{f}_{X_i}(x)\mathrm{d}{x}{\setlength\arraycolsep{0.5pt}=}1{\setlength\arraycolsep{0.5pt}-}\!\!\int_0^\infty e^{-x^\beta P_i^{-1}\tau\sigma^2}\prod_{j=1}^3\mathcal{L}_{I_j}\left[x^\beta P_i^{-1}\tau\right]\widehat{f}_{X_i}(x)\mathrm{d}{x}.
\end{align}
Similar to (\ref{laplace1_2}), the Laplace transform of the interference resulted from the first tier becomes
\begin{align}
&\mathcal{L}_{I_1}\left[x^\beta P_i^{-1}\tau\right]=e^{-\pi\lambda'_1\left(\frac{P_1}{P_i}\right)^{\frac{2}{\beta}}x^2\mathcal{Z}_2(a)}.
\end{align}
We then have the theorem. {\hfill\ensuremath{\blacksquare}}
\subsection{Proof of Theorem 6}\label{sec:prooftheroem6}
Accordingly, we have
\begin{align}\label{52}
\mathcal{P}_{3,j}=1-\!\!\int_0^\infty\!\!\!\int_0^\infty \mathbb{P}\left[\text{SINR}_j(y)\begin{array}{c|}
                          \!\!>\tau
                         \end{array}~x,y\right]f_{X_1,Y_j}(x,y)\mathrm{d}{x}\mathrm{d}{y},
\end{align}
where
\begin{align}\label{53}
\!\!\!\mathbb{P}\left[\text{SINR}_j(y)\begin{array}{c|}
                          \!\!>\tau
                         \end{array}~x,y\right]
&=\mathbb{P}\left[\begin{array}{c|}
                           \!\!g_{j,0}>y^\beta P_j^{-1}I_r\tau
                         \end{array}~x,y\right]=e^{-\frac{y^\beta W\tau}{L_0P_j}}\prod_{i=1}^3\mathcal{L}_{I_i}\left[y^\beta P_j^{-1}\tau\right].
\end{align}
Meanwhile, we have $\mathcal{L}_{I_i}[y^\beta P_j^{-1}\tau]=\text{exp}[{-\pi\lambda_i(\frac{P_i}{P_j})^{\frac{2}{\beta}}y^2\mathcal{Z}_1(\tau)}],i=2,3$, and $\mathcal{L}_{I_1}[y^\beta P_j^{-1}\tau]=\text{exp}[{-\pi\lambda'_1y^2(\frac{P_1}{P_j}\tau)^{\frac{2}{\beta}}\int_{x^2(y^\beta\frac{P_1}{P_j}\tau)^{-\frac{2}{\beta}}}^\infty\frac{1}{1+u^{\frac{\beta}{2}}}\mathrm{d}{u}}]$.
Then we obtain (\ref{oi_z_3}) and the proof is finished. \hfill\ensuremath{\blacksquare}

\end{appendix}
\bibliographystyle{IEEEtran}
\bibliography{paper}

\end{document}